\documentclass[letterpaper,12pt]{article}

\usepackage[utf8]{inputenc}
\usepackage{geometry}
\usepackage{setspace}

\usepackage{pdfpages}
\usepackage[title]{appendix}
\usepackage{soul}
\usepackage{blkarray}

\usepackage{booktabs,caption,widetable,threeparttable}
\usepackage{tikz,fullpage}
\usetikzlibrary{arrows,	petri, topaths}
\usepackage{rotating,graphicx,lscape,xr-hyper,pgf,float,color}
\usepackage[position=top]{subfig}
\usepackage{nicematrix}

\usepackage[authoryear]{natbib}
\usepackage{bibentry}
\usepackage{url}
\usepackage[colorlinks = true, allcolors = blue]{hyperref}
\usepackage[inline]{enumitem}
\usepackage{listings}
\usepackage[reftex]{theoremref}
\usepackage{amsmath, amsthm, amssymb, bm, mathtools, amsfonts, subdepth}
\mathtoolsset{showonlyrefs}

\theoremstyle{remark}
\newtheorem{rem}{Remark}

\theoremstyle{definition}
\newtheorem{example}{Example}
\newtheorem{definition}{Definition}
\newtheorem{innercustomthm}{Example}
\newenvironment{customthm}[1]
{\renewcommand\theinnercustomthm{#1}\innercustomthm}
{\endinnercustomthm}

\theoremstyle{plain}
\newtheorem{thm}{Theorem}
\newtheorem{ass}{Assumption}
\newtheorem{lemma}{Lemma}

\usepackage{fixme}

\DeclarePairedDelimiter{\abs}{\lvert}{\rvert}
\DeclarePairedDelimiter{\norm}{\lVert}{\rVert}

\newcommand{\E}{\mathbb{E}}
\newcommand{\V}{\mathbb{V}}

\DeclareMathOperator{\tr}{tr}

\DeclareMathOperator{\plim}{plim}

\title{Estimation and exclusion restrictions in clustered \\ linear models\thanks{
We are grateful to Isaiah Andrews, Patrik Guggenberger, Lihua Lei, and Morten Ø. Nielsen for helpful and insightful discussions. We thank the authors of \cite{egger2022general} for sharing access to their data.
This research was supported by grants from the Danish National and Aarhus University Research Foundations (DNRF grant number DNRF186 and AUFF grant number AUFF-E-2022-7-3)
}}

\author{
    \textsc{Anna Mikusheva}%
    \thanks{
        Department of Economics, M.I.T.,  and Aarhus Center for Econometrics, 50 Memorial Drive, E52-526, Cambridge, MA, 02142, United States. E-mail: amikushe@mit.edu.; corresponding author
      }, \ 
    \textsc{Mikkel S\o lvsten}%
    \thanks{
         Department of Economics and Aarhus Center for Econometrics, Aarhus University, Fuglsangs Allé 4, Building 2621, B 13a, 8210 Aarhus V, Denmark. E-mail: miso@econ.au.dk. },
         \ 
         \textsc{Baiyun Jing}%
    \thanks{Department of Economics, Harvard University, Littauer Center, Cambridge, MA, 02138, United States. E-mail: bjing@fas.harvard.edu}}

\date{ August 11, \the\year }

\begin{document}

\maketitle

\begin{abstract}
    We study linear regression models with clustered data, high-dimensional controls, and intricate exclusion restrictions. We propose a correctly centered internal instrument IV estimator that accommodates a broad class of exclusion restrictions and allows within-cluster dependence. The estimator admits a simple leave-out interpretation and is computationally tractable. We derive a central limit theorem for the associated quadratic form and propose a robust variance estimator. We also develop identification-robust inference procedures. Our framework extends dynamic panel methods to general clustered settings. We illustrate the approach in a large-scale fiscal intervention in rural Kenya, where spatial interference generates the exclusion-restriction pattern.

    \bigskip
    
    \noindent
    \textsc{Keywords:}  weak exogeneity, many controls, interference, OLS inconsistency, internal instrument, leave-out approach 
     
     \bigskip
     
     \noindent
     \textsc{JEL Codes:} C13, C22
\end{abstract}

\section{Introduction}
Clustered datasets---such as panel, network, spatial, or group-structured data in which individuals or firms are observed repeatedly or nested within groups---have become increasingly common in empirical research. These settings present several methodological challenges. One is the need to accommodate heterogeneity, often addressed by including fixed effects and flexible time or group trends, which typically requires many control variables. A more fundamental difficulty is within-cluster dependence: observations in the same cluster may be correlated due to, e.g., spatial and network interference, spillover effects, and time-series dependence. Such dependencies complicate statistical inference and limit the effective aggregation of information across observations. \looseness=-1

The concept of exogeneity becomes more nuanced in the presence of clustered data. As we demonstrate, assuming only a per-observation exclusion restriction---uncorrelatedness of the error with the regressor within an observation---may deliver no identifying variation for consistent estimation of structural parameters. By contrast, strict exogeneity---each error term uncorrelated with all regressors in the cluster---is often implausible in many empirical contexts. \looseness=-1

To bridge strict exogeneity and unrestricted dependence, we assume the error term is uncorrelated with a subset of within-cluster regressors; the subset is application-specific. In panel data, it is common to assume errors are mean-zero conditional on current and past (but not future) regressors, allowing feedback from shocks to future policies. In spatial contexts, localized spillovers motivate analogous restrictions: regressors corresponding to sufficiently distant units within a cluster may plausibly be uncorrelated with a given error term, even if nearby regressors are not. Examples include spatial leakage in deforestation \citep{jayachandran2017cash}, spillovers in policing \citep{blattman2021place}, and fiscal intervention experiments in rural Kenya \citep{egger2022general}.
In network settings, treatment effects may propagate along social ties, so regressors associated with nodes that are unconnected—or sufficiently distant in the network—may be uncorrelated with a given error term \citep{paluck2016changing}. These partial exogeneity assumptions more accurately reflect the dependence structures present in many empirical applications.\looseness=-1

When only a subset of exclusion restrictions holds, it becomes inappropriate to treat regressors as fixed or to condition on them. Their stochastic variation may correlate with the errors, undermining standard arguments that rely on conditioning on regressors. As a result, even estimators like OLS must be interpreted as ratios of stochastic quadratic forms, which introduces several challenges. \looseness=-1

The most well-known is \emph{Nickell bias} \citep{nickell1981biases}, which occurs when the expected value of the OLS error numerator is nonzero, leading to asymptotic bias. While prominent in dynamic panels with fixed effects and lagged outcomes, analogous bias arises more broadly under clustered dependence.  
A second challenge is inference: standard variance estimators often fail to account for dependence among cross-cluster terms in the quadratic form. Finally, the denominator may remain asymptotically random, unlike in classical settings, leading to \emph{weak identification} and further complicating inference. \looseness=-1

This paper studies the estimation of structural parameters in linear regression models with clustered data, high-dimensional controls, and researcher-specified exclusion restrictions. A key feature of the framework is the grouping of observations into disjoint clusters, permitting within-cluster dependence and independence across clusters. The setting accommodates unbalanced panels with multiple high-dimensional fixed effects, as well as spatial, network, and group-structured data. Our approach extends dynamic panel methods to a broader class of models, addressing the central challenges of bias, inference, and identification. \looseness=-1

Our first contribution is to characterize a class of \textit{correctly centered} internal instrument estimators that remove the leading asymptotic bias of OLS. The approach accommodates high-dimensional exogenous controls and adapts to the exclusion structure specified in the application. The resulting estimator is an internal instrument IV estimator, defined as the one closest to OLS in a specified norm. The estimator is easy to implement and asymptotically efficient under some conditions. \looseness=-1

The procedure has a simple interpretation: for each observation, controls are partialled out using only those observations whose errors are uncorrelated with the observation's regressor, yielding an observation-specific leave-out projection. A just-identified IV regression is then performed on the transformed equation, using the original regressor as the instrument. \looseness=-1

We further characterize the efficiency loss from varying the strength of exclusion restrictions. In particular, assuming exogeneity only at the observation level eliminates all identifying variation when fixed effects are present. This observation underscores the importance of carefully specifying the exogeneity structure in applied work. \looseness=-1

Our second contribution addresses inference. We show that, outside special cases such as (i) strict exogeneity and (ii) block-diagonal residualization (e.g., models with only cluster fixed effects), the estimator's numerator is a nontrivial quadratic form in the errors. In these more general settings, standard cluster-robust variance estimators may fail, and valid inference requires central limit theorems for clustered quadratic forms. These issues are especially acute in models with many high-dimensional controls, such as two-way fixed effects \citep{verdier2018estimation}. We propose a new central limit theorem and variance estimator that account for this structure. \looseness=-1

Our third contribution addresses weak identification, which arises when the internal instrument---though valid---captures little identifying variation due to many controls or weak exclusion restrictions. As in the classical dynamic panel literature \citep{blundell1998initial, bun2010weak}, weak identification can compromise standard inference. We develop inference procedures that remain valid even when the estimator's denominator exhibits substantial sampling variability. \looseness=-1

We propose a comprehensive empirical strategy that combines a correctly centered estimator with identification-robust inference and confidence sets. We apply it to a prominent randomized evaluation of a large-scale fiscal intervention in rural Kenya \citep{egger2022general}. A key challenge in this setting is spatial interference: treatment in one village affects outcomes in neighboring villages, complicating the choice of plausible exclusion restrictions. We show that weakening the exclusion restrictions yields less precise estimates and wider confidence sets. \looseness=-1

This paper contributes to several strands of the econometrics literature. It relates to the extensive work on linear dynamic panel data models---a mature field with too many important contributions to list exhaustively. For recent overviews on correcting Nickell bias, see the surveys by \cite{okui20211} and \cite{bun2015dynamic}. Our estimation approach is more closely aligned with the internal instrument framework developed in \cite{anderson1981estimation}, \cite{arellano1991some}, \cite{ahn1995efficient}, \cite{alvarez2003time}, and \cite{blundell1998initial}, and we extend the framework to accommodate more general models. In addition, we address weak identification concerns noted in the dynamic panel literature, particularly in \cite{bun2010weak}. We also contribute to the literature on quadratic central limit theory for statistics that combine linear and quadratic components under clustered sampling. Closely related work by \cite{ligtenberg2023inference} develops such theory for many- and weak-instrument inference in models with clustered data and external instruments. Our setting instead concerns internal-instrument estimation with clustered data and researcher-specified exclusion restrictions, but gives rise to a related technical problem. \looseness=-1


The remainder of the paper is organized as follows. Section~\ref{sec: setup} describes the data structure, a plausible set of exclusion restrictions, and two modeling perspectives---outcome-based and design-based. Section~\ref{sec: removing bias} characterizes correctly centered internal instrument estimators, introduces our proposed estimator as the solution to an efficiency optimization problem, and interprets it as a leave-out internal instrument procedure. Section~\ref{sec: variance of numerator} addresses uncertainty quantification, and Section~\ref{sec: CLT} establishes a central limit theorem for quadratic forms with clustered data. Section~\ref{sec: weak id} develops inference procedures robust to weak identification and a corresponding variance estimator. Section~\ref{sec: empirical application} illustrates the full strategy in an empirical application, showing how both the estimator and, especially, its uncertainty depend on the maintained set of assumptions.  \looseness=-1

\paragraph{Notation.}  
We use \(0< c < 1 \) and \( C > 0 \) to denote generic finite constants that may vary across equations but do not depend on the sample size. For a matrix \( A \), we let \( A^+ \) denote its Moore–Penrose generalized inverse, \( \tr(A) \) its trace, \( \|A\|_F^2 = \tr(A'A) \) its Frobenius norm squared, and \( \|A\| \) its operator (spectral) norm. \looseness=-1

\section{Setup and Main results}\label{sec: setup}

We consider estimation of a structural parameter $\beta$ in the linear regression model
\begin{align}\label{eq: panel model}
y_\ell = x_\ell \beta + w_\ell'\delta + e_\ell, \quad \ell = 1, \dots, n.
\end{align}
Here, $\ell$ indexes the $n$ observations. The parameter of interest, $\beta$, is scalar, although most results extend to fixed-dimensional vectors. The vector of strictly exogenous controls $w_\ell$ has dimension $K$, which may be large but satisfies $K<n$. Let $W=[w_1,\dots,w_n]'$ denote the non-random control matrix. We assume throughout that $W$ has full rank and define  $M = I_n - W(W'W)^{-1}W'$.\looseness=-1

Unlike in standard cross-sectional settings, the data are partitioned into $N$ disjoint clusters, with independence across clusters and arbitrary dependence within clusters. 
\begin{ass}\label{ass: clustered}
    Let \(\{ S_i \}_{i=1}^N\) be a partition of the index set \(\{1, \dots, n\}\), so that \(\cup_{i=1}^N S_i = \{1, \dots, n\}\) and \(S_i \cap S_j = \emptyset\) for \(i \ne j\). The set $\{e_\ell, x_\ell\}_{\ell \in S_i}$ is independent across $i$.
\end{ass}
 Let $i(\ell)$ denote the cluster containing observation $\ell$, and let $T_i=|S_i|$, so that $n=\sum_{i=1}^N T_i$. Assume that the random regressors $x_\ell$ and errors $e_\ell$ have uniformly bounded second moments. Several exclusion restrictions may be imposed in regression \eqref{eq: panel model}. The weakest is contemporaneous exogeneity,
$
\E[x_\ell e_\ell]=0=\E[e_\ell], \ell=1,\dots,n,
$
which defines $\beta$. Under independent observations and mild regularity conditions, this restriction ensures consistency of ordinary least squares (OLS). With clustered data and cluster fixed effects, however, contemporaneous exogeneity alone need not ensure consistency.\looseness=-1

At the other extreme is \emph{strict exogeneity},
$\E[e_\ell\mid x]=0,
 x=[x_1,\dots,x_n]',$
which implies $\E[x_{\tilde\ell}e_\ell]=0$ for all $\tilde\ell$ and $\ell$. Under strict exogeneity, OLS is unbiased conditional on $x$. This condition is often too strong in clustered settings, where unmodeled dependence between regressors and outcomes within clusters is likely. Exclusion restrictions should instead reflect the institutional or structural features of the application. \looseness=-1

We therefore allow the researcher to specify an $n\times n$ indicator matrix $\mathcal{E}$ encoding the imposed moment restrictions. Specifically, $\mathcal{E}_{\tilde\ell\ell}=1$ denotes the restriction
$\E[x_{\tilde\ell}e_\ell]=0,$
whereas $\mathcal{E}_{\tilde\ell\ell}=0$ permits $\E[x_{\tilde\ell}e_\ell]\ne0$. Independence across clusters implies $\E[x_{\tilde\ell}e_\ell]=0$ whenever $i(\ell)\ne i(\tilde\ell)$. For convenience, we set $\mathcal{E}_{\tilde\ell\ell}=1$ in such cases. Hence, zeros in $\mathcal{E}$ can occur only within cluster blocks and indicate the absence of an exclusion restriction. We study estimation of $\beta$ in \eqref{eq: panel model} under the restrictions encoded by $\mathcal{E}$. \looseness=-1

\begin{ass}\label{ass: panel model}
Assume $\E[e_\ell]=0$, $\mathcal{E}_{\ell\ell} = 1$ for all $\ell$, and $\E[x_{\tilde\ell}e_\ell ] = 0$ whenever $\mathcal{E}_{\tilde\ell\ell} = 1$.  
\end{ass}

Assumption~\ref{ass: panel model} ensures contemporaneous exogeneity, \( \E[ x_\ell e_\ell] = 0 \), which underlies the desirable properties of OLS in standard cross-sectional settings.

\begin{example}[Network data]\label{ex: new ex- network} Consider the randomized controlled trial studied in \citep{paluck2016changing}, in which students are randomly assigned to participate, $x_\ell$, in a program that teaches teenagers conflict-resolution skills.   A researcher may worry about interference and believe that the treatment affects not only the outcomes of treated students, but also those of their close friends. In this case, it is reasonable to assume that an exclusion restriction holds only for pairs of observations that are not directly connected by friendship, so that the matrix $\mathcal{E}_{\ell\tilde\ell}=\mathbb{I}\{\ell \mbox{ and } \tilde \ell\mbox{ are not friends}\}$ is given by the complement of the friendship adjacency matrix. Such spillovers of information or resources through social networks are increasingly common in modern randomized controlled trials. Prominent examples include \cite{banerjee2013diffusion}, \cite{banerjee2019using}, and \cite{banerjee2024can}. \looseness=-1\qed
\end{example}

\begin{example}[Spatial data]\label{ex: ex 1new}
Suppose a researcher is interested in estimating the average treatment effect of a policy using a randomized controlled trial, in which an economic resource \(x_\ell\) is randomly assigned across municipalities $\ell$.  Assume that municipalities are nested within metropolitan areas indexed by \(i\), which serve as the clusters.     Spatial interference arises when outcomes in a given location may also be affected by treatment in neighboring locations—for example, because of resource sharing or price adjustments in shared local markets. Such spillover effects violate the exogeneity assumption across neighboring localities, since treatment in one location directly affects outcomes in the neighboring locality. If the researcher believes that interference occurs only when locations are within $R$ kilometers of each other and is negligible at larger distances, then the matrix $\mathcal{E}_{\ell\tilde\ell}=\mathbb{I}\{d_{\ell\tilde\ell}>R\}$ encodes indicators of whether two locations lie more than $R$ kilometers apart. Notable examples of such settings are spatial leakage in deforestation \citep{jayachandran2017cash}, spillovers in policing  \citep{blattman2021place}, or a fiscal intervention experiment in rural Kenya \citep{egger2022general}. We discuss the latter in Section\ref{sec: empirical application}.  \looseness=-1\qed
\end{example}

\begin{example}[Unbalanced dynamic panel data]\label{ex: ex 2}
Suppose a researcher is interested in estimating the effect of program participation \( x_{it} \) on student achievement \( y_{it} \) using a regression with individual and time fixed effects: $y_{it} = \alpha_i + \mu_t + \beta x_{it} + e_{it}.$  It is plausible that current achievement \( y_{it} \) could influence future eligibility or desire to participate \(x_{is}\) for $s>t$, thereby inducing correlation between the current error term \( e_{it} \) and future regressors \( x_{is} \), violating strict exogeneity. In such settings, it is common to assume \emph{weak exogeneity}: \( \E[x_{is} e_{it} ] = 0 \) for \( s \le t \), meaning the error is uncorrelated with current and past regressors but may influence future regressors. 
This structure arises naturally when, for example, $x_{it}=y_{i,t-1}$, or more generally when regressors and outcomes evolve jointly over time. If observations are ordered by cluster and time, $\mathcal{E}$ then contains triangular within-cluster blocks. Although this setting is covered by a large dynamic-panel literature, we include it to connect our results to that literature.
\looseness=-1\qed
\end{example}

\subsection{Summary of the results and empirical recipes}

Consider estimation of $\beta$ in \eqref{eq: panel model} under Assumption~\ref{ass: panel model}. Our first result is negative: OLS can exhibit substantial bias and may be inconsistent.

To illustrate, consider Example~\ref{ex: new ex- network}, where the researcher observes a friendship adjacency matrix and imposes exclusion only between students who are not direct friends. We simulate data from a stochastic block model in which clusters represent school classes and within-cluster edges represent friendships. Treatment $x_\ell$ is randomly assigned, while the outcome includes spillovers from friends' treatments weighted by unobserved friendship intensities. Appendix~\ref{sec:sim} describes the design.

 \begin{figure}[htbp!]            \begin{center}              \includegraphics[width=0.5\linewidth]{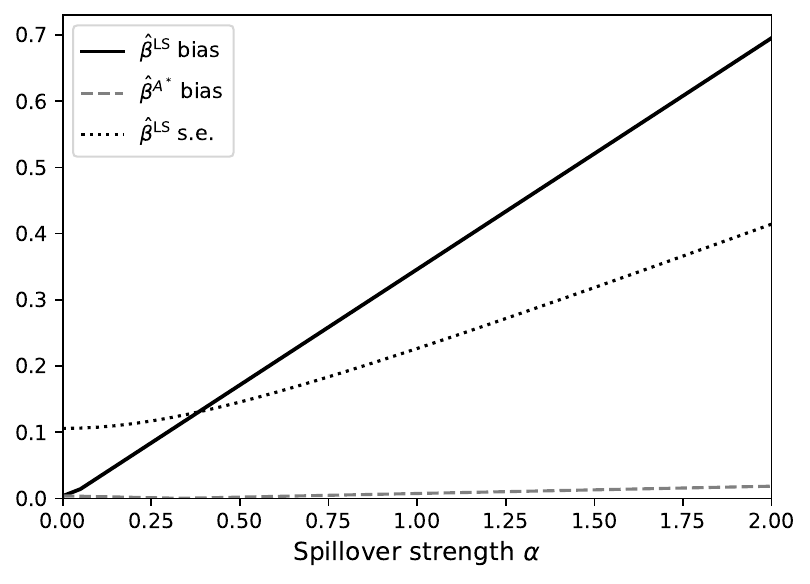}
      \caption{Bias of $\hat\beta^{\mathrm{LS}}$ in simulations}
      \label{fig:sim}
      \end{center}
      {\small Notes. This figure presents the absolute bias of $\hat\beta^{\mathrm{LS}}$ (solid line) and its cluster-robust standard error (dotted line) in a design-based DGP with $n=500$ and varying spillover strength $\alpha$. The absolute bias of 
  the proposed estimator $\hat\beta^{A^*}$ (described in Theorem \ref{thm: leave out}, dashed line) is shown as a benchmark.}
  \end{figure}

  Figure~\ref{fig:sim} illustrates OLS bias in the network-interference setting of Example~\ref{ex: new ex- network}. The horizontal axis measures the strength of interference and hence the degree of violation of strict exogeneity. Despite random treatment assignment, OLS can be severely biased, with bias exceeding its reported standard error. The bias arises from within-cluster spillovers, as removing cluster fixed effects mixes up the direct effect on a person's own outcome with indirect effects on a friend's outcome. An analogous phenomenon in the dynamic-panel setting of Example~\ref{ex: ex 2} is the well-known Nickell bias.

We propose a new estimator that addresses this problem and can be implemented in two steps. First, for each observation $\ell$, residualize the outcome and regressor with respect to the controls using only observations $\tilde\ell$ for which $\mathcal{E}_{\ell\tilde\ell}=1$. In Example~\ref{ex: new ex- network}, this means using only non-friends. Second, estimate a just-identified IV regression of the residualized outcome on the residualized regressor, using the original regressor $x_\ell$ as the instrument and including no controls. Figure~\ref{fig:sim} shows that this estimator removes the OLS bias.

More generally, we characterize a broad class of estimators that eliminate this bias, which we call \emph{correctly centered}. This class corresponds to the estimator class of \cite{arellano1991some} in dynamic panel models. By construction, our estimator also avoids the many-instrument bias often associated with that approach.

A common empirical practice combines OLS with cluster-robust standard errors. Standard justifications rely heavily on strict exogeneity, which we replace with the weaker restrictions encoded by $\mathcal{E}$. As shown in Section~\ref{sec: variance of numerator}, when strict exogeneity fails and partialling out controls induces leakage across clusters---as with unbalanced two-way fixed effects or multiple demographic controls---conventional cluster-robust standard errors may understate uncertainty.

To address this problem, we establish a new quadratic central limit theorem for clustered data. The result shows that conventional cluster-robust variance estimators omit relevant quadratic terms. We therefore propose a jackknife procedure that yields asymptotically valid confidence sets for $\beta$.

Here is a simple recipe for constructing the confidence interval. Our estimator is $\hat\beta=\frac{x'A^*y}{x'A^*x}$, where $A^*$  is defined in Theorem~\ref{thm: leave out}. The confidence set contains all values $\beta_0$ for which $AR(\beta_0)=\frac{(x'A^*(y-\beta_0x))^2}{V(\beta_0)}$ is below the corresponding $\chi_1^2$ critical value. The jackknife variance is $V(\beta_0)=\sum_{i=1}^N(\mathcal{Z}-\mathcal{Z}_{(i)})^2$ where $\mathcal{Z}=x'A^*(y-\beta_0x)$, $\mathcal{Z}_{(i)}=x_{(i)}'A^*(y-\beta_0x)_{(i)}$ and variables indexed by $(i)$ are set to zero for observations in cluster $i$. The confidence set may be obtained either by grid inversion or analytically, since the resulting inequality is quadratic in $\beta_0$.

Different exclusion restrictions yield different estimators: stronger assumptions generally improve efficiency at the cost of robustness. This makes the choice of restrictions important and, in our view, one that should be guided by economic and institutional knowledge. With nested exclusion restrictions, validity of both sets yields two consistent estimators and permits a $J$-test of the more restrictive specification. We caution, however, against using such tests to select exclusion restrictions, for the same reason that responsible empirical practice cautions against selecting instruments based on $J$-tests.

\subsection{Asymptotic bias of OLS}
When strict exogeneity fails, the OLS estimator may suffer from asymptotic bias and, in some cases, may even be inconsistent. \looseness=-1

\begin{lemma}\label{lem: nickell bias}
    Suppose assumptions \ref{ass: clustered}, \ref{ass: panel model} and \ref{ass: technical} (stated in the Appendix) hold.  Then,
    \begin{align}
        \hat \beta^\mathrm{LS} = \beta + \frac{\frac{1}{n} \sum_{\ell=1}^n \sum_{\tilde\ell=1}^n M_{\tilde\ell\ell} \E[x_{\tilde \ell} e_\ell  ]}{Q} + o_p(1),
            \end{align}
    where $M_{\tilde \ell\ell}$ are entries of the $n\times n$ projection matrix $M$, and $Q=\plim_{n \rightarrow \infty} \frac{1}{n}\sum_{\ell=1}^n\tilde x_{\ell}^2$.
\end{lemma}

This lemma is a relatively straightforward extension of the well-known Nickell bias \citep{nickell1981biases}. Nickell’s original work focused on a more restricted setting commonly referred to as a “dynamic panel data” model. \looseness=-1

\begin{example} [Dynamic panel data] 
    \begin{align}\label{eq: dynamic panel data}
    y_{it}=\alpha_i+\beta y_{i,t-1}+e_{it}.
\end{align}
This is a special case of our Example \ref{ex: ex 2} with the weak exogeneity restriction $\E[e_{it} \!\mid\!\{ y_{is} : s < t \}]=0$ for all $i,t$. Consider the simplest case of a balanced panel with $T_i = T$ for all $i$. The projection matrix $M$ corresponds to demeaning within units, and thus has a block-diagonal structure when observations are ordered by cluster. Under homoskedasticity with error variance $\sigma^2$, $\E[x_{\tilde\ell} e_\ell] = \E[y_{is} e_{it}] = \beta^{s - t} \sigma^2$ when $s \ge t$. Substituting these expressions into Lemma \ref{lem: nickell bias} yields the bias formula originally derived in \cite{nickell1981biases}. It is well known that in this setting, OLS is inconsistent for fixed $T$, as the bias does not vanish. The leading term of the bias is $O(1/T)$, and if $T$ grows slower than the number of clusters $N$, the bias may dominate the variance asymptotically, resulting in invalid inference. \qed \looseness=-1
\end{example}
 
The source of Nickell bias is well understood: removing fixed effects mixes observations across periods, while future regressors may correlate with current errors. Our framework is more general than the standard dynamic panel specification in \eqref{eq: dynamic panel data}. Importantly, the bias characterized by \cite{nickell1981biases}, and generalized in Lemma~\ref{lem: nickell bias}, is an asymptotic bias: subtracting it from OLS yields a consistent (or asymptotically unbiased) estimator. This differs from finite-sample bias because it ignores randomness in the OLS denominator, which vanishes asymptotically but remains present in finite samples.\looseness=-1

\subsection{Design based modeling}
\label{sec:design_base}
Model \eqref{eq: panel model} combined with Assumption~\ref{ass: panel model} can be viewed as an instance of \emph{outcome modeling}, since it specifies the outcome equation. In cross-sectional settings, however, OLS estimators perform well when either the outcome or the treatment equation is correctly specified. The alternative approach—based on the treatment equation—is referred to as \emph{design-based} and is commonly used when researchers understand the random assignment of treatment, such as in randomized controlled trials (RCTs). In this paper, we define a design-based model by specifying the following treatment equation, along with Assumption~\ref{ass: design model}:
\begin{align}\label{eq: design model}
x_\ell = w_\ell^\prime \delta_x + v_\ell.
\end{align}

\begin{ass}\label{ass: design model}
Assume \( \E[v_\ell] = 0 \), \( \mathcal{E}_{\ell\ell} = 1 \) for all \( \ell \), and \( \E[v_{\tilde\ell}(y_\ell - \beta x_\ell)] = 0 \) whenever \( \mathcal{E}_{\tilde\ell\ell} = 1 \).
\end{ass}

\begin{customthm}{\ref{ex: new ex- network}}[continued]
    Suppose a researcher wishes to estimate the effect of access to resources on individual outcomes by conducting an RCT in which individuals \( \ell \) are randomly provided with additional resources. The researcher assumes that the treatment assignment equation is correctly specified as $x_\ell = w_\ell' \delta_x + v_\ell,$ where \( w_\ell' \delta_x \) denotes the propensity score, which may depend on individual characteristics as well as cluster fixed effects.\looseness=-1

Now suppose there are reasons to expect \emph{interference}—that is, an individual's outcome may be influenced by the treatment assignment of their friends, for example, due to the sharing of resources within social networks. The researcher observes the friendship network and randomizes treatment independently of it. The outcome equation is posited as
\[
y_\ell = \beta x_\ell + g(W, F_\ell' x) + e_\ell,
\]
where \( F_\ell \) encodes the friendship links relevant for individual $\ell$, so that $F_\ell' x$ is the vector of treatment assignments for peers who may influence individual $\ell$. 
The error term \( e_\ell \) is assumed to be independent of all assignments \( x \), but may be correlated within clusters to capture common shocks (e.g., village-level effects). In this case, the outcome modeling approach is complicated by the need to specify the functional form of the indirect effect. This effect might depend on whether \emph{any} friends are treated, the \emph{number} of treated friends, or the unobserved \emph{intensity} of friendships. The quantity \( g(W, F_\ell' x) \) may vary systematically within clusters, and this variation may not be well captured by fixed effects or observed covariates. Consequently, if the indirect effect is absorbed into the unmodeled regression error term, the outcome model \eqref{eq: panel model} may be misspecified.\looseness=-1

Nonetheless, design-based assumptions remain valid for pairs \( \ell \) and \( \tilde\ell \) who are not friends:
\[
\E[v_{\tilde\ell}(y_\ell - \beta x_\ell)] = \E[v_{\tilde\ell}(g(W,F_\ell' x) + e_\ell)] = \E[v_{\tilde\ell}] \cdot \E[g(W,F_\ell' x) + e_\ell] = 0,
\]
by the independence of treatment assignment.  \qed
\end{customthm}

The asymptotic bias expression from Lemma~\ref{lem: nickell bias} can equivalently be restated in the design-based framework:
\begin{align}
\hat\beta^\mathrm{LS} = \frac{x' M y}{x' M x} = \frac{v'M y}{v'  Mx} = \beta + \frac{v'M (y-\beta x)}{v' M x}=\beta + \frac{\frac{1}{n} \sum_{\ell,\tilde\ell=1}^n  M_{\tilde\ell\ell} \E[v_{\tilde \ell} (y_\ell-\beta x_\ell)  ]}{Q} + o_p(1).
\end{align}

The design-based perspective  is useful in settings where the treatment assignment mechanism is credible but the outcome equation may be difficult to specify. The primary exposition in this paper is framed within the outcome modeling approach. However, all results can equivalently be reformulated under the design-based perspective with only minor modifications. We provide such restatements in the remarks following the main theorems. Additionally, where appropriate, we discuss the possibility of a \emph{doubly robust} formulation, in which the researcher assumes that either the outcome equation or the treatment equation is correctly specified, without committing to which one. 

\section{Correctly centered estimators}\label{sec: removing bias}

\subsection{Class of  estimators}

This subsection shows that unbiased estimation is generally impossible once regressors are random, even under standard orthogonality conditions, and motivates a weaker requirement that we call correct centering. Assume we have  data \( \{y_\ell, x_\ell\} \) generated according to model \eqref{eq: panel model} under Assumption~\ref{ass: panel model}. The set of controls \( W \) and the exclusion restriction matrix \( \mathcal{E} \) are given to the econometrician and treated as fixed. All constructive results in this section do not require cluster structure of Assumption~\ref{ass: clustered}, while the negative result in Lemma~\ref{lem- no unbiased est} holds both with or without Assumption~\ref{ass: clustered}. We consider the problem of estimating the parameter \( \beta \). Let \( \mathcal{F} \) denote the class of all joint distributions \( F \) over \( (x, y) \) that satisfy model \eqref{eq: panel model} under Assumption~\ref{ass: panel model}.
\begin{lemma}\label{lem- no unbiased est}
There exists no estimator \( \hat\beta = u(x, y) \) that is unbiased for all \( F \in \mathcal{F} \).
\end{lemma}

In fact, the proof of Lemma \ref{lem- no unbiased est} establishes something even stronger: in a cross-sectional setting with an independent sample $\{y_i, x_i\}_{i=1}^n$ from a correctly specified regression model $y_i = \beta x_i + e_i$, where the regressor $x_i$ is random and $\E[e_i] = \E[x_i e_i] = 0$, no unbiased estimator of $\beta$ exists. The OLS estimator is not unbiased in this setting because the randomness of the regressors induces randomness in the OLS denominator, making it impossible to pass the expectation operator through the ratio. This observation motivates the need to consider a broader and more appropriate class of estimators, in which the normalization is separated from the centering.


\begin{definition}
An estimator \( \hat\beta = \frac{C_1(x, y)}{C_2(x)} \) is said to be \emph{correctly centered} if for all \( F \in \mathcal{F} \):
\[
\mathbb{E}_F[C_1(x, y)] = \beta \, \mathbb{E}_F[C_2(x)].
\]
\end{definition}
\looseness=-1
Correct centering is a weaker condition than unbiasedness. The key distinction lies in the treatment of randomness in the denominator. When the denominator is non-random—so that $\mathbb{E}[C_2(x)] = C_2(x)$—the two concepts coincide. In a correctly specified cross-sectional regression with random regressors, OLS is correctly centered but not unbiased.

For many of the estimators we consider, we can establish a form of the law of large numbers (under broadly applicable technical conditions) that ensures $C_2(x) - \mathbb{E}[C_2(x)] \xrightarrow{p} 0$ and $C_1(x, y) - \mathbb{E}[C_1(x, y)] \xrightarrow{p} 0$—with any necessary normalization absorbed into the functions themselves. For instance, Lemma \ref{lem: nickell bias} establishes that $C_2(x) = \frac{1}{n}x'Mx = \frac{1}{n}\mathbb{E}[x'Mx] + o_p(1)$ and $C_1(x, y) = \frac{1}{n}\mathbb{E}[x'My] + o_p(1)$. Regarding the denominator, what matters is the size of its uncertainty relative to its signal—that is, the strength of identification. If $\frac{C_2(x)}{\mathbb{E}[C_2(x)]} \xrightarrow{p} 1$ and the above law of large numbers holds, then correct centering implies consistency (or asymptotic unbiasedness) of the estimator.

\begin{lemma}\label{lem: OLS not cc} If there exists a pair with
$\mathcal E_{\tilde\ell\ell}=0$ and $M_{\tilde\ell\ell}\neq0$, the OLS estimator is not correctly centered in the outcome model \eqref{eq: panel model} with Assumption \ref{ass: panel model} or in the design-based model \eqref{eq: design model} with Assumption~\ref{ass: design model}.   
\end{lemma}

To better understand the concept, it is helpful to recall that in a standard cross-sectional overidentified IV regression, the two-stage least squares (2SLS) estimator is \textit{not} correctly centered. When the number of instruments is large, the resulting bias can be of first-order importance and is known as the many-instrument problem. That is, most feasible estimators suggested in \cite{arellano1991some} for dynamic panel date are not correctly centered and in practice demonstrate high biases. In contrast, a one-step just-identified IV estimator—as well as jackknife or leave-one-out IV estimators in overidentified cross-sectional settings—remains correctly centered whenever the instruments are exogenous.

\begin{lemma}\label{lem: correct centered class} Consider an estimator
\(
\hat{\beta}^A = \frac{x' A y}{x' A x},
\)
where \( A \) is an \( n \times n \) fixed matrix.\footnote{If \(\mathcal F\) is restricted to distributions also satisfying Assumption~\ref{ass: clustered}, not every correctly centered estimator has the form \(\hat{\beta}^A\), or is even linear in \(y\). 
The independence between clusters imposes strong restrictions on the joint distribution of the data and allows higher-order $U$-statistics to be added to the numerator. 
For instance, if \((x_i,y_i), \ i=1,\dots,4\) are observations from distinct (and thus independent) clusters, and \(y_i = \beta x_i + e_i\), then $\mathbb{E}\big[x_1 x_2 y_3 y_4 - y_1 y_2 x_3 x_4\big] = 0.$ This observation connects to recent discussions in \cite{hansen2022modern} and \cite{portnoy2022linearity,lei2022estimators}.} It is correctly centered if and only if 
\( A \)  satisfies the following conditions:
\begin{equation}\label{eq: restrictions}
\emph{(POP) } AM = A 
\quad \text{and} \quad 
\emph{(CC) } A_{\tilde \ell \ell} = 0 \text{ for all pairs } (\tilde \ell, \ell) \text{ such that } \mathcal{E}_{\tilde\ell\ell} = 0.
\end{equation}
\end{lemma}


Here, (POP) stands for the \emph{partialling-out property}, and (CC) refers to \emph{correct centering}. Under condition (POP), the estimator satisfies:
\[
\hat\beta^A = \frac{x' A y}{x' A x} = \frac{x' A (\beta x + W \delta + e)}{x' A x} = \beta + \frac{x' A e}{x' A x}.
\]
Condition (CC) guarantees that \( \mathbb{E}[x' A e] = 0 \), ensuring correct centering of the estimator.
\looseness=-1

The class of estimators described in Lemma~\ref{lem: correct centered class} coincides with the class of just-identified linear IV estimators that use linear internal instruments. Specifically, define the instrument \( z = A'x \), a linear transformation of the regressors. Condition (CC) implies the exogeneity condition \( \mathbb{E}[z_\ell e_\ell] = 0 \) for all \( \ell \). If we estimate model \eqref{eq: panel model} via a just-identified IV regression using \( z \) as the instrument for \( x \), and including \( W \) as controls, then a straightforward generalization of the Frisch-Waugh-Lovell theorem for IV estimators yields:
\[
\hat\beta^{\mathrm{IV}} = \frac{z' M y}{z' M x} = \frac{x' A M y}{x' A M x} = \frac{x' A y}{x' A x} = \hat\beta^A,
\]
where the final equality follows from the partialling-out property (POP).
\looseness=-1

\begin{rem}
    If the data are generated from the design-based model \eqref{eq: design model} under Assumption~\ref{ass: design model}, then the results of Lemma~\ref{lem: correct centered class} continue to hold with a slightly modified (POP) condition. Specifically, we must replace (POP) with the following: (POP$^\prime$) $MA=A$. In the case of the doubly robust model, the corresponding condition becomes: (POP$^{\prime\prime}$) $MA=AM=A$. \qed
\end{rem}


Our approach is closely related to internal-instrument methods for dynamic panels, including \cite{anderson1981estimation}, \cite{arellano1991some}, and \cite{ahn1995efficient}. These methods typically stack available exclusion restrictions and use overidentified 2SLS or GMM. We instead form a just-identified linear combination of internal instruments and estimate by one-step IV, yielding a correctly centered estimator. Overidentified procedures generally lose this property because estimating the weighting matrix adds randomness, which can generate substantial many-instrument bias when overidentification is large \citep{alvarez2003time}. Different matrices $A$ satisfying (CC) and (POP) span a broad class of linear internal instruments, including the infeasible optimally weighted instrument of \cite{arellano1991some}. \looseness=-1

\subsection{Asymptotic efficiency considerations}

Let $\mathcal{A}$ denote the set of $n \times n$ matrices satisfying conditions (POP) and (CC). This set forms a linear subspace within the space of $n \times n$ matrices. The full space of $n \times n$ matrices has dimension $n^2$. Condition (POP), which is equivalent to $AW = 0_{n \times K}$, imposes $nK$ linear constraints. Condition (CC) imposes an additional $L$ zero restrictions, where $L$ is the number of entries $A_{\tilde\ell \ell}$ required to be zero due to (CC). Therefore, the dimension of $\mathcal{A}$ is at least $n(n - K) - L$. If the dimension of $\mathcal{A}$ is positive, then our estimator class contains infinitely many candidates, raising the natural question: Which matrix $A$ should we use? As a default, we propose the choice
\begin{equation}\label{eq: optimization problem}
    A^*=\arg\min_{A\in \mathcal{A}}\|A-I_n\|_F=\arg\min_{A\in \mathcal{A}}\|A-M\|_F.
\end{equation}
The equality holds due to condition (POP). The Frobenius norm defines a Euclidean geometry on the space of matrices, with the corresponding inner product given by $\langle A, B \rangle_F = \tr(A'B)$. The optimization problem \eqref{eq: optimization problem} admits a unique solution: the orthogonal projection of $M$ onto the subspace $\mathcal{A}$ with respect to the Frobenius inner product. In the lemma below, we show that this choice of $A^*$ can be motivated by asymptotic efficiency considerations.

\begin{lemma}\label{lem: efficiency statement}
    Suppose Assumption~\ref{ass: panel model} holds and assume that $\E[ee' \!\mid\!x] = \sigma^2 I_n$ and $\E[xx'] = \lambda I_n$ for some constants $\sigma^2,\lambda > 0$. Consider the function    
    $$
    V(A)=\frac{\V(x'Ae)}{\left(\E[x'Ax]\right)^2}=\frac{\V(z'e)}{\left(\E[z'x]\right)^2}
    $$
    where $z = A' x$. Then the minimum of $V(A)$ over the set $\mathcal{A}$ is attained at $A^*$. 
\end{lemma}

The function $V(A)$ represents the asymptotic variance of the IV estimator under strong identification, i.e., when ${x'Ax}/{\E[x'Ax]} \xrightarrow{p} 1$. This efficiency criterion follows the classical approach of \cite{chamberlain1987asymptotic}: conditional homoskedasticity, $\E[ee' \!\mid\!x] = \sigma^2 I_n$implies that the optimal instrument $z^* $ minimizes the first stage prediction error  over all $z = A'x$ with $A \in \mathcal{A}$: 
\begin{equation}\label{eq: alternative optimization}
(A^*)'x=z^*=\arg\min_{z=A'x, A\in \mathcal{A}}\E\!\left[ \|z-x\|^2 \right].  
\end{equation}
The assumption $\E[xx'] = \lambda I_n$  ensures the equivalence between the optimization problems \eqref{eq: optimization problem} and \eqref{eq: alternative optimization}.\looseness=-1

We do not take Lemma~\ref{lem: efficiency statement} too literally, since homoskedasticity and homogeneity are restrictive assumptions. Rather, it motivates a default correctly centered estimator, much as OLS is a natural default in linear regression. Efficiency can be improved by imposing additional structure, at the cost of robustness or efficiency elsewhere in the parameter space. Any prior information about $x$---such as its scale or dynamic dependence---can, in principle, be used to improve efficiency, for example through cluster-specific weighting or time transformations as in GLS. The key restriction is that $A$ must not depend on the realized randomness in either $x$ or $e$, as doing so can introduce bias, such as the many-instrument bias of two-step GMM \citep{arellano1991some}. If additional structure is known, for example $\E[xx']=\Omega$, then asymptotic efficiency is attained by solving:
$$
A_{\Omega}^* = \arg\min_{A\in \mathcal{A}}\tr\left((A-I_n)'\Omega(A-I_n)\right) = \arg\min_{A\in \mathcal{A}} \|\Omega^{1/2}(A-I_n)\|_F.
$$
Most results in this paper extend naturally to this weighted setting, except for the leave-one-out projection characterization discussed below in Theorem 1\ref{thm: leave out3}.\looseness=-1

\subsection{Solution to the optimization problem}

This subsection derives and discusses the solution to the optimization problem \eqref{eq: optimization problem}.\looseness=-1

We begin by introducing what we refer to as leave-out projections. For each index $\tilde{\ell}$, define $\mathcal{E}_{\tilde{\ell},\cdot}$ to be the $\tilde{\ell}$th row of the matrix $\mathcal{E}$, which contains indicators for those indices $\ell$ such that $\E[x_{\tilde{\ell}} e_\ell] = 0$. In other words, $\mathcal{E}_{\tilde{\ell},\cdot}$ identifies all observations for which $x_{\tilde{\ell}}$ is uncorrelated with their error term. Let $\mathbf{i}_{\tilde{\ell}} = \operatorname{diag}(\mathcal{E}_{\tilde{\ell},\cdot})$ denote the $n \times n$ diagonal selection matrix with ones on the diagonal corresponding to entries where $\mathcal{E}_{\tilde{\ell} \ell} = 1$, and zeros elsewhere. Then define
$W_{\tilde{\ell}} = \mathbf{i}_{\tilde{\ell}} W,$
an $n \times K$ matrix that retains the rows $w_\ell$ for which $\mathcal{E}_{\tilde{\ell} \ell} = 1$, and sets the remaining rows to zero. The corresponding leave-out projection matrix is given by $M^{(\tilde\ell)}=I-W_{\tilde \ell}(W_{\tilde \ell}'W_{\tilde \ell})^{+}W_{\tilde \ell}',$ which partials out controls using only observations whose errors are uncorrelated with $x_{\tilde\ell}$. Unlike the full projection $M$, it excludes data points that violate exogeneity.\looseness=-1

\begin{thm}\label{thm: leave out}
    Assume $W$ has full rank. Let $A^*$ be the solution to the optimization problem \eqref{eq: optimization problem}. Then, $A^*$ admits the following three equivalent characterizations:
    \begin{enumerate}[label=(\roman*)]
    
        \item\label{thm: leave out1} Let the linear restrictions in (CC) and (POP) be written as $\mathcal{L}vec(A)=0$, then
        $$ vec(A^*)=(I_{n^2}-\mathcal{L}'(\mathcal{L}\mathcal{L}')^+\mathcal{L})vec(M).
        $$
        
        \item\label{thm: leave out2} There exists an $n\times n$ matrix $B$ satisfying condition $B_{\tilde\ell\ell}=0$ whenever $ \mathcal{E}_{\tilde\ell\ell}=1$ such that
        $A^*=(I_n-B)M$. 
        
        \item\label{thm: leave out3}$A^*_{\tilde\ell\ell}=M^{(\tilde\ell)}_{\tilde\ell\ell}$, if $\mathcal{E}_{\tilde\ell\ell}=1$, and  $A^*_{\tilde\ell\ell}=0$, otherwise.

    \end{enumerate}
    
\end{thm}
Characterization \ref{thm: leave out1} in Theorem \ref{thm: leave out} highlights the analytical tractability of $A^*$ by noting its connection with linear projection.\looseness=-1

Characterization \ref{thm: leave out2} shows that $A^*$ closely resembles the OLS projection matrix $M$, with the difference given by $A^* - M = -BM$, where $B$ is a sparse matrix. Specifically, $B_{\tilde\ell\ell}$ is nonzero only if $\E[x_{\tilde\ell} e_\ell] \ne 0$. As a result, $B$ is block-diagonal (with blocks corresponding to clusters) and has zeros on the main diagonal. Thus, $A^*$ can be computed separately within each cluster, which is computationally appealing. \looseness=-1

Finally, part \ref{thm: leave out3} of Theorem \ref{thm: leave out} offers an intuitive interpretation of the solution. For each observation $\tilde\ell$, we perform an observation-specific transformation of the original model \eqref{eq: panel model} by partialling out the control $w_{\tilde\ell}$ using only observations uncorrelated with $x_{\tilde\ell}$.  Define the residualized outcome as $ y_{\tilde\ell}^*=y_{\tilde\ell}-w_{\tilde\ell}^\prime\hat\delta^y_{{\tilde\ell}}$, where $\hat\delta^y_{{\tilde\ell}}$ is obtained from regressing $y$ on $w$ using only observations with errors uncorrelated with $x_{\tilde\ell}$. Similarly define $x_{\tilde\ell}^*$ and $e_{\tilde\ell}^*$. These definitions yield the transformed equation: \looseness=-1
\begin{equation}\label{eq: differenced equation}
    y_{\tilde\ell}^*=\beta x_{\tilde\ell}^*+e_{\tilde\ell}^*.
\end{equation}
Since $e^*_{\tilde\ell}$ is constructed using only  observations uncorrelated with $x_{\tilde\ell}$, we have $\E[x_{\tilde\ell}e_{\tilde\ell}^*]=0$. Hence, $x_{\tilde\ell}$ is a valid instrument in \eqref{eq: differenced equation}, and the IV estimator is $\hat\beta=\frac{\sum_{\tilde\ell}x_{\tilde\ell}y^*_{\tilde\ell}}{\sum_{\tilde\ell}x_{\tilde\ell}x^*_{\tilde\ell}}=\frac{x'A^*y}{x'A^*x}=\hat\beta^{A^*}$.\looseness=-1

\begin{example}[Panel data with fixed effects]\label{ex: 4}
    Consider the model 
    $$
    y_{it}=\alpha_i+\beta x_{it}+e_{it}, \mbox{   with    } \E[e_{it} \!\mid\!x_{it},x_{i,t-1},\dots]=0,
    $$
    where the error is unpredictable given current and past values of the regressor, but may affect future values. In this case, the optimal matrix $A^*$ has a block-diagonal structure, with each block corresponding to a cluster. Our estimator transforms the data by demeaning $y_{it}$ and $x_{it}$ using only current and future observations:
\[y_{it}^*=y_{it}-\frac{1}{T_i-t+1}\sum_{s=t}^{T_i}y_{is}, ~~~ x_{it}^*=x_{it}-\frac{1}{T_i-t+1}\sum_{s=t}^{T_i}x_{is}\]
and estimates the regression $y_{it}^* = \beta x_{it}^* + e_{it}^*$ using $x_{it}$ as an instrument. This estimator lies within the class proposed by \cite{arellano1991some}, but applies more broadly. See also \cite{hayashi1983nearly,arellano1995another}.\looseness=-1 \qed
\end{example}


\paragraph{Price of robustness.} 
Our estimator can be seen as a robust alternative to OLS when strict exogeneity, $\E[e_\ell \!\mid\!x] = 0$, is not credible. The exclusion matrix $\mathcal{E}$ determines the types of violations it is robust to:  Example \ref{ex: 4} allows future dependence. A natural question is whether this robustness reduces efficiency when strict exogeneity holds. Under the assumptions of Lemma \ref{lem: efficiency statement}, we have
$
V(A^*)=\frac{\sigma^2}{\lambda\tr(A^*)}=\frac{\sigma^2}{\lambda\|A^*\|_F^2},
$
where $\tr(A^*)$ reflects the effective sample size.
In the standard panel model with strict exogeneity and homoskedasticity, OLS achieves the efficiency bound with $\tr(M)=\sum_{i=1}^N(T_i-1)=n-N$, as one degree of freedom is lost per cluster due to demeaning. Each block $M_{ii}$ contributes $\tr(M_{ii}) = T_i - 1$. In Example \ref{ex: 4}, $\tr(A_{ii})=T_i-1-\sum_{t=2}^{T_i}\frac{1}{t}$. With $\sum_{t=2}^{T_i}\frac{1}{t} \ge \log( T_i+1) - \log(2)$, this shows an efficiency loss when guarding against dependence on any future regressor. Finally, in a panel model with fixed effects, assuming contemporaneous exogeneity only $\E[x_\ell e_\ell] = 0$ leads to $A^* = 0$. In this case, the model contains no usable identifying variation. \looseness=-1

\paragraph{Design-based model and double robustness.} 
Theorem~\ref{thm: leave out} gives the solution when the outcome equation \eqref{eq: panel model} is correctly specified. If instead the design equation \eqref{eq: design model} for $x_\ell$ is correctly specified, we impose (POP$^\prime$), $MA=A$, with minor modifications to the theorem. In part~\ref{thm: leave out2}, $A^*=M(I_n-B)$. In part~\ref{thm: leave out3}, construct a proxy for $v_\ell$ as the $\ell$th residual from \eqref{eq: design model}, estimated using only observations $\tilde\ell$ with $\mathcal{E}_{\tilde\ell\ell}=1$, and use this residual as an instrument in a just-identified IV regression of $y_\ell$ on $x_\ell$.

A doubly robust estimator is correctly centered when either the outcome equation \eqref{eq: panel model} or the design equation \eqref{eq: design model} is correctly specified. This requires (POP$^{\prime\prime}$), $MA=AM=A$, and modifies part~\ref{thm: leave out2} to $A^*=M(I_n-B)M$. No analogue of the simple construction in part~\ref{thm: leave out3} is available. If $W$ includes cluster fixed effects, the doubly robust estimator is invariant to their values in both $x_\ell$ and $y_\ell$. One-sided estimators are invariant only to fixed effects in one variable. They remain correctly centered under the corresponding valid assumption, but their variance may depend on the fixed effects in the other variable.

\section{Quantification of uncertainty}\label{sec: variance of numerator}

\subsection{Importance of cross-dependence}

This section relies heavily on Assumption~\ref{ass: clustered} on cluster independence. Let $x_i$ denote the $T_i \times 1$ vector of regressor values within cluster $S_i$---that is, all $x_\ell$ for which $i(\ell) = i$---and define $e_i$ analogously. Due to the cluster structure, the pairs $(x_i, e_i)$ are independent across $i$. Let $A_{ij}$ denote the $T_i \times T_j$ blocks of the matrix $A\in\mathcal{A}$, corresponding to observations in clusters $S_i$ and $S_j$, respectively. The numerators of the estimation errors take the quadratic form:
$$
\quad
\hat\beta^A-\beta=\frac{x'Ae}{x'Ax} 
\quad \text{and} \quad
x'Ae=\sum_{i,j}x_i'A_{ij}e_j.
$$
As these are quadratic forms in random $(x_i,e_i)$, the variance of the estimators has a non-standard structure. In particular, establishing asymptotic normality may require a CLT for quadratic forms. \looseness=-1

\paragraph{Special case: block-diagonal structure.} It is worth highlighting a special case ---where the numerator reduces to linear forms. If the off-diagonal blocks $A_{ij}$ are equal to zero for all $i\neq j$, then $x'Ae=\sum_i x_i'A_{ii}e_i$.
Since $(x_i, e_i)$ are independent across clusters, standard inference follows directly from the law of large numbers and central limit theorem, provided the number of clusters $N \to \infty$ and that the clusters satisfy a negligibility condition. This scenario arises in  Example \ref{ex: 4}  when only cluster fixed effects $\alpha_i$ are included in $W$. In that case, usual clustered standard errors can be used. However, this structure breaks down when more complex controls are included. \looseness=-1

\begin{example}[Multi-way fixed effects]\label{ex: 6}
 Consider the model:
\begin{equation}
    y_{it}=\alpha_i+\gamma_{g(i,t)}+\beta x_{it}+e_{it},
\end{equation}
where $\alpha_i$ is an individual fixed effect, $g(i,t)$ denotes the group (e.g., teacher or classroom) assigned to individual $i$ at time $t$, and $\gamma_g$ is a group fixed effect. For instance, $y_{it}$ may represent student $i$’s academic achievement in year $t$, and $x_{it}$ could be last year’s score (to study persistence of achievement) or an intervention determined by past performance. The fixed effect $\alpha_i$ captures student ability, while $\gamma_g$ represents teacher or classroom effects. The clustering is assumed at the individual level, and we assume weak exogeneity: $\E[e_{it} \!\mid\!\{x_{is}, s\leq t\}]=0$. \looseness=-1

 Assume all students are observed in two periods. In year 1, students are split evenly between teachers $g = 1$ and $g = 2$; in year 2, they are randomly reassigned to $g = 3$ and $g = 4$. Each classroom has the same number of students.
For a student $i$ with teachers $(1,3)$, the differenced equation is:
$$
\Delta y_{i}=y_{i2}-y_{i1}=\gamma_3-\gamma_1+\beta\Delta x_i+\Delta e_i.
$$
OLS partials out teachers' effects and transforms this to:
$$
y_i^*=\Delta y_i-\left\{\frac{3}{4}\overline{\Delta y}_{(1,3)}+\frac{1}{4}\overline{\Delta y}_{(1,4)}+\frac{1}{4}\overline{\Delta y}_{(2,3)}-\frac{1}{4}\overline{\Delta y}_{(2,4)}\right\},
$$
where the term in brackets estimates $\gamma_3 - \gamma_1$, and $(2,3)$ refers to the group of students who had teacher $g = 2$ in year 1 and $g = 3$ in year 2, and the upper bar denotes the average over the group.  Similar transformations apply to $x_i^*$ and $e_i^*$. Our estimator $\hat\beta^{A^*}$ is IV on the equation $y_i^* = \beta x_i^* + e_i^*$, using $x_{i1}$ as an instrument:
$\hat\beta^{A^*}-\beta=\frac{\sum_{i=1}^N x_{i1}e_i^*}{\sum_{i=1}^Nx_{i1}x_i^*}$. But since $e_i^*$ depends on all $e_{jt}$, the terms $x_{i1} e_i^*$ are dependent across $i$. Here, $A^*$ is not block-diagonal, making standard inference approaches invalid and motivating the need for a different variance estimator. \qed\looseness=-1
\end{example}

Example \ref{ex: 6} highlights several key points. First, the difference between $e_i^*$ and $e_i$ reflects estimation error from recovering the teacher effects $\gamma_g$. Usage of the full sample to estimate these effects mixes observations and induces complex dependence, complicating inference. If $\gamma_g$ estimates are very precise, then $\sum_i x_{i1} e_i^*$ will be close to $\sum_i x_{i1} e_i$, and ignoring the dependence may still yield valid inference.\looseness=-1

Second, the OLS transformation for a student in group $(1,3)$ involves not only students with shared teachers ($g=1$ or $g=3$), but also those indirectly connected, e.g., students in $(2,4)$, who share a teacher with others in $(1,4)$ or $(2,3)$. More generally, we can view students as nodes in a network, where edges represent shared teachers. OLS demeaning propagates information across this network, creating strong dependencies within connected components and complicating standard inference methods.\footnote{While OLS demeaning is efficient under the benchmark conditions of Lemma \ref{lem: efficiency statement}, one may prefer a less efficient estimator that induces less dependence. In Example \ref{ex: 6}, instead of using the full sample to estimate $\gamma_3 - \gamma_1$, one could rely only on students with the same teacher history, e.g., $\overline{\Delta y}_{(1,3)}.$ This results in independent “super-clusters” defined by the teacher pairs: $(1,3)$, $(1,4)$, $(2,3)$, and $(2,4)$. One can go further by forming pairs of students with identical teacher histories and differencing them, creating many small, independent clusters, then using clustered standard errors on that level. However, this comes at a cost: the resulting estimator may be highly inefficient due to poor estimation of fixed effects. This strategy is explored in \cite{verdier2018estimation}. \looseness=-1}

\subsection{The correct formula for quantifying uncertainty }

Denote $\lambda_i = \E[x_i]$ and $v_i = x_i - \lambda_i$. Then,
$$
x'Ae=\sum_{j=1}^N\left(\sum_{i=1}^N\lambda_i^\prime A_{ij}+v_j'A_{jj}\right)e_j+\sum_{j=1}^N\left(\sum_{i\neq j}v_i'A_{ij}\right)e_j=\sum_{j=1}^N\omega_j e_j+\sum_{j=1}^NQ_je_j.
$$
where $Q_j = \sum_{i \ne j} v_i' A_{ij}$ is mean-zero and independent of $(e_j, v_j)$. The above is the so-called Hoeffding decomposition \citep{hoeffding1948class}; the two sums are the first two U-statistics, which are uncorrelated with each other by construction. Thus, the variance decomposes as: \looseness=-1
$$
\V(x'Ae)=\V\left(\sum_{j=1}^N\omega_je_j\right)+\V\left(\sum_{j=1}^N Q_je_j\right).
$$
The first term is the first order U-statistic with $\V\left(\sum_{j=1}^N\omega_je_j\right)=\sum_{j=1}^N\V\left(\omega_je_j\right)$. The second term has correlated summands:
\begin{align}
\label{eq: quad_var}
    \V\left(\sum_{j=1}^N Q_je_j\right) = \sum_{j=1}^N \V\left(Q_je_j\right) + \sum_{j=1}^N\sum_{i\neq j}\E[v_i'A_{ij}e_jv_j'A_{ji}e_i].
\end{align}
Standard cluster-robust variance formulas omits the final term $\sum_{j=1}^N\sum_{i\neq j}\E[v_i'A_{ij}e_jv_j'A_{ji}e_i]$, which captures cross-cluster dependence. This term vanishes in two important cases: (i) under strict exogeneity, $\E[e_j v_j'] = 0$ for all $j$; or (ii) when $A$ is block-diagonal, i.e., $A_{ij} = 0$ for $i \ne j$. These are the settings where standard clustered standard errors yield valid inference. \looseness=-1

\section{Asymptotic Gaussianity}\label{sec: CLT}

We establish a result on the asymptotic Gaussianity of the numerator of the estimation error, $x'Ae$. This result enables the construction of a hypothesis test for $H_0: \beta = \beta_0$, as well as the development of a confidence set for $\beta$ in the presence of clustered data. Implementation details are provided in Section~\ref{sec: weak id}. Subsection~\ref{subsec: nontech} presents a non-technical exposition and discusses the assumptions required to establish Gaussianity. Subsection~\ref{subsec: tech} is intended for theoretical econometricians and contains the technical details. Readers whose focus is more applied may skip it.

\subsection{Non-technical exposition}\label{subsec: nontech}

Two main challenges arise when establishing Gaussianity. First, the statistic includes not only a linear term but also a quadratic form in random shocks. Second, intra-cluster dependence introduces complexities, raising the question of how to restrict such dependence and what trade-offs arise between dependence and cluster size. Our asymptotic framework requires the number of clusters to grow, i.e., $N \to \infty$.

The first challenge is addressed by developing a new Central Limit Theorem (Lemma \ref{lem: CLT for quadratics}) for quadratic forms of independent, normalized random vectors. These vectors represent observations within each cluster. The theorem establishes that, under certain negligibility conditions, a quadratic form involving matrix-weighted sums of such vectors converges in distribution to a normal distribution.  The core requirement is that the contribution of any single cluster—as measured, for instance, by the Frobenius norm of the corresponding submatrix—becomes negligible relative to the total variance. A key strength of this result is that it avoids direct control over cluster sizes; the cluster size enters only implicitly through restrictions on the weight matrices.

Related joint limit theory is developed by \cite{ligtenberg2023inference} for clustered IV statistics combining linear and bilinear terms. Both approaches handle recentering after partialling out many controls and remove products with nonzero expectations; \cite{ligtenberg2023inference} drops all within-cluster products, whereas our exclusion and partialling-out restrictions allow some to be retained. His CLT uses growing-rank projection weights and restricts maximum cluster size relative to rank; ours permits general weights and instead imposes covariance and matrix-norm conditions, yielding related but non-nested results.

Our main result---Gaussianity of $x'Ae$---is formalized in Theorem~\ref{thm: gaussianity} for a general matrix $A$ and later specialized to our proposed estimator matrix $A^*$. Theorem~\ref{thm: gaussianity} combines a Lyapunov-type CLT for linear forms with the new quadratic CLT. Two sets of assumptions are required.

The first set (Assumption~\ref{ass: gaussianity}) imposes conditions on the distribution of errors, including moment bounds and tail behavior within clusters. Specifically, any normalized linear combination of errors within a cluster (i.e., one with unit variance) must have a uniformly bounded fourth moment. Likewise, normalized combinations of $e_i v_i'$ must have uniformly bounded $(2 + 2\delta)$ moments. These assumptions ensure that intra-cluster dependence is well-approximated by the covariance matrix and rule out extreme tail dependence, such as aligned outliers not reflected in the covariance. Further, we require that certain correlations are bounded away from one and that variances are bounded away from zero, thereby excluding cases of near-perfect predictability.

The second set (Assumption~\ref{ass: gaussianity matrix A}) concerns the interaction between the weight matrix $A$ and the operator norms of the cluster-specific covariance matrices $\Sigma_i$. These assumptions ensure that each cluster's contribution to the total variance remains asymptotically negligible. By working with $\|\Sigma_i\|$, we allow for a trade-off between cluster size and intra-cluster dependence. Two stylized examples clarify the implications for the quadratic form (Assumption~\ref{ass: gaussianity matrix A}\ref{ass: gaussianity matrix A3}) that typically requires stronger restrictions :

\begin{itemize}
      \item In the presence of a pervasive cluster-level shock (e.g., a factor structure), errors within a cluster may be strongly correlated. Then $\|\Sigma_i\|$ can grow proportionally with cluster size $T_i$,  up to $\|\Sigma_i\| \leq C T_i$. To achieve Gaussianity of the quadratic term in this case, stronger conditions must be imposed on the weight matrix and cluster size, such as $\max_i T_i^4 / n \to 0$.
       
    \item At the opposite extreme, if no pervasive shocks exist and dependence decays rapidly with distance (in some metric) within the cluster, then $\|\Sigma_i\|$ may remain uniformly bounded regardless of cluster size. In such cases, weaker conditions, such as $\max_i T_i^2 / n \to 0$  suffice.       
\end{itemize}

\subsection{Technical details}\label{subsec: tech}

\begin{lemma}\label{lem: CLT for quadratics}
Let $\xi_i$ be independent $T_i \times 1$ random vectors with $\E[\xi_i] = 0$, $\V(\xi_i) = I_{T_i}$, and $\E|\tau' \xi_i|^4 \leq C \| \tau \|^4$ for all $T_i \times 1$ vectors $\tau$. Let $\Gamma$ be an $n \times n$ block matrix with blocks $\Gamma_{ij}$ of size $T_i \times T_j$, symmetric across the diagonal ($\Gamma_{ij} = \Gamma_{ji}'$), and with $\Gamma_{ii} = 0$ for all $i$. Assume: \looseness=-1
$$ \frac{\max_i\left(\sum_{j\neq i}\|\Gamma_{ij}\|_F^2\right)}{\sum_i\sum_{j\neq i}\|\Gamma_{ij}\|_F^2}\to 0\mbox{ 
     and     }\frac{\|\Gamma\|^2}{\|\Gamma\|_F^2}\to 0.$$ 
    Then as $N\to\infty$
    $$\frac{1}{\sqrt{V}}\sum_{i=1}^N\sum_{j\neq i} \xi_i'\Gamma_{ij}\xi_j\xrightarrow{d} N(0,1),~~~ \mbox{where  
   }  V=2\|\Gamma\|_F^2=2\sum_i\sum_{j\neq i}\|\Gamma_{ij}\|_F^2.
    $$
\end{lemma}

Let $\Sigma_i$ denote the covariance matrix of $(e_i', v_i')'$, with submatrices $\Sigma_{e,i}$ and $\Sigma_{v,i}$. Define the normalized shock vector for cluster $i$ as $\xi_i = \Sigma_i^{-1/2} (e_i', v_i')'$, which is a $2T_i \times 1$ vector. Also define the cross-covariance matrix $\Psi_i = \E[e_i (v_i \otimes e_i)']$ of dimension $T_i \times T_i^2$, and the variance matrix $\Phi_i = \V(v_i \otimes e_i)$ of dimension $T_i^2 \times T_i^2$. \looseness=-1

\begin{ass} \label{ass: gaussianity} 
The cluster-specific errors $(e_i', v_i')'$ satisfy the following conditions:\looseness=-1
    
\begin{enumerate}[label=(\roman*)]
    \item\label{ass: gaussianity1} 
    Each $\Sigma_i$ is invertible, with $\|\Sigma_i^{-1}\| \leq C$, and $\xi_i$ are independent, mean-zero vectors satisfying $\E[v_i' A_{ii} e_i] = 0$.\looseness=-1        
       
    \item\label{ass: gaussianity2} 
    Variance conformity: uniformly over clusters,
    $
c\bigl(\Sigma_{v,i}\otimes\Sigma_{e,i}\bigr)
\preceq \Phi_i
\preceq
C\bigl(\Sigma_{v,i}\otimes\Sigma_{e,i}\bigr).
$

   \item\label{ass: gaussianity3} 
   Bounded alignment:
        $\max_i \left\| \Sigma_{e,i}^{-1/2} \Psi_i \Phi_i^{-1/2} \right\| \leq 1 - c;$    \looseness=-1
    
    \item\label{ass: gaussianity4} 
    Higher-order moment bounds:
       $\E |\tau' \xi_i|^4 \leq C \| \tau \|^4,$ and $
    \E (\tau' (v_i \otimes e_i))^{2 + 2\delta} \leq C \left( \tau' \Phi_i \tau \right)^{1 + \delta},
    $
    for some $\delta > 0$.\looseness=-1
\end{enumerate}
\end{ass}

\noindent The assumptions above ensure two linear CLTs: one for sums of $e_i$ and one for sums of $v_i \otimes e_i$. Condition \ref{ass: gaussianity2} guarantees that the covariance matrix of $(\Sigma_{v,i}^{-1/2} v_i) \otimes (\Sigma_{e,i}^{-1/2} e_i)$ has eigenvalues bounded above and away from zero, preventing both near-degeneracy and disproportionately large product variance in any direction. Condition \ref{ass: gaussianity3} ensures that correlations between $e_i$ and $v_i \otimes e_i$ remain bounded away from one. Condition  \ref{ass: gaussianity4} rules out heavy-tailed behavior and strong higher-order dependence that could dominate the covariance-based approximation.\looseness=-1

\begin{ass}\label{ass: gaussianity matrix A}
Let $A$ be an $n \times n$ matrix such that
\begin{enumerate}[label=(\roman*)]
    \item\label{ass: gaussianity matrix A1}$\sum_{i=1}^N(A'\lambda)_i^\prime\Sigma_{e,i}(A'\lambda)_i+\| A\|_F^2\geq cn$;\looseness=-1

    \item\label{ass: gaussianity matrix A2} $\frac{\max_i\|\Sigma_i\|^2\|A_{ii}\|^2_F}{n}\to 0$ and $\frac{\max_i(A'\lambda)_i^\prime\Sigma_{e,i}(A'\lambda)_i}{n}\to 0$ ;\looseness=-1
         
    \item\label{ass: gaussianity matrix A3} either one of the two conditions hold: 
    \begin{itemize}
        \item[(a)] $ \frac{\max_i\|\Sigma_i\|^2\cdot\sum_{i=1}^N\sum_{j\neq i}\|A_{ij}\|_F^2}{n}\to 0$
        \item[(b)] $\frac{\max_i\|\Sigma_i\|^2\cdot\max_i\sum_{j=1}^N( \| A_{ij}\|_F^2+\| A_{ji}\|_F^2)}{n}\to 0$ and $\frac{\| A\|^2\cdot\max_i\|\Sigma_i\|^2}{n}\to 0$.
    \end{itemize}
        
\end{enumerate}
\end{ass}

Part \ref{ass: gaussianity matrix A1} guarantees that the normalization in the CLT is no slower than $1/\sqrt{n}$. Part \ref{ass: gaussianity matrix A2} ensures cluster negligibility for the linear CLT, Part \ref{ass: gaussianity matrix A3} (b) does so for the quadratic CLT, while Part \ref{ass: gaussianity matrix A3} (a) assumes that the quadratic part is negligible in comparison to the linear one. \looseness=-1

\begin{thm}\label{thm: gaussianity}
Let $x = \lambda + v$, where $\lambda = \E[x]$, and suppose Assumptions~\ref{ass: gaussianity} and~\ref{ass: gaussianity matrix A} hold. Then, as $N \to \infty$, $\frac{x'Ae}{\omega} \xrightarrow{d} \mathcal{N}(0,1),$ where  $\omega^2 = \V(x'Ae) \geq c n $ for some  $c > 0.$
\end{thm}

We now discuss how Assumption~\ref{ass: gaussianity matrix A} applies to the matrix $A^*$ defined in Theorem~\ref{thm: leave out}. According to Lemma~\ref{lem: efficiency statement}, the Frobenius norm $\|A^*\|_F^2$ captures the effective sample size. Assumption~\ref{ass: gaussianity matrix A}\ref{ass: gaussianity matrix A1} holds provided $A^*$ retains sufficient variation relative to $n$. From Theorem~\ref{thm: leave out}(c), each row of $A^*$ corresponds to a projection, and we have: 
$$\sum_{\tilde\ell=1}^n(A_{\ell\tilde\ell}^*)^2=A_{\ell\ell}^*\mbox{ 
 and }\|A^*\|_F^2=\tr(A^*)=\sum_{\ell}A_{\ell\ell}^*.$$
Since $A^*$ is observable, $\|A^*\|_F^2$ can be computed and directly compared to $n$. A simple sufficient condition for Assumption~\ref{ass: gaussianity matrix A}\ref{ass: gaussianity matrix A1} is illustrated below:
\begin{example}
Suppose the data is indexed by $\ell = \ell(i, t)$ and that a standard weak exogeneity condition holds: $\E[e_{it} \!\mid\! x_{is} : s \le t] = 0$. Let $\mathcal{X}_0 = \{\ell(i,1)\}$ denote the set of first-period observations for each cluster. For these $\ell$, we have $\mathcal{E}_{\ell \tilde\ell} = 1$ for all $\tilde\ell$, so no observations are excluded from the projection, and thus $A^*_{\ell \ell} = M_{\ell \ell}$. If $\min_{\ell \in \mathcal{X}_0} M_{\ell \ell} > c > 0$ and average cluster size is bounded, then Assumption~\ref{ass: gaussianity matrix A}\ref{ass: gaussianity matrix A1} holds. The condition $M_{\ell \ell} > c$ is standard in the many regressors/instruments literature \citep[e.g.,][]{cattaneo2018inference}. \qed
\end{example}

\begin{ass}[Balanced design]\label{ass: balanced} Assume that the set of controls and the exclusion matrix $\mathcal{E}$ are such that for any non-trivial submatrix $M_{\ell}$---formed from entries $M_{\ell_1\ell_2}$ with $\mathcal{E}_{\ell\ell_1}=0$ and $\mathcal{E}_{\ell \ell_2}=0$ ---satisfies $\|M_{\ell}^{-1}\|\leq C$.    
\end{ass}

\noindent Theorem~\ref{thm: leave out} shows that $A^* = (I - B)M$, where $B$ is a sparse matrix with relatively few nonzero entries. Assumption~\ref{ass: balanced} ensures that $A^*$ remains close to $M$ under appropriate matrix norms.\looseness=-1

\begin{lemma}\label{lem: assumption check} The following primitive conditions are sufficient for parts of Assumption~\ref{ass: gaussianity matrix A}:
\begin{itemize}
    \item If $\frac{\max_i\|\Sigma_i\|^2T_i}{n}\to 0,$ then the first part of Assumption \ref{ass: gaussianity matrix A}\ref{ass: gaussianity matrix A2} holds;
    \item If Assumption \ref{ass: balanced} holds, $\|\lambda\|_\infty<\infty$ and $\frac{(\max_i T_i)^3\|\Sigma\|\cdot\|M\|_{\infty}^2}{n}\to 0$, then the second part of Assumption \ref{ass: gaussianity matrix A}\ref{ass: gaussianity matrix A2} holds;
    \item If Assumption \ref{ass: balanced} holds, and $\frac{\max_i\|\Sigma_i\|^2\cdot\max_i T_i^2}{n}\to 0,$ then  Assumption \ref{ass: gaussianity matrix A}\ref{ass: gaussianity matrix A3}(b) holds.
\end{itemize}
\end{lemma}

\noindent For linear CLTs, the condition $\max_i T_i^2 / n \to 0$ suffices \citep{hansen2019asymptotic, sasaki2022non}, while the quadratic CLT (Assumption~\ref{ass: gaussianity matrix A}\ref{ass: gaussianity matrix A3}(b)) requires the stricter condition $\max_i T_i^4 / n \to 0$.\looseness=-1

In high-dimensional settings with many regressors, an analog of the second part of Assumption~\ref{ass: gaussianity matrix A}\ref{ass: gaussianity matrix A2}—e.g., $\frac{\|M \lambda\|_\infty^2}{n} \to 0$—is often imposed as a high-level assumption because direct verification is difficult; see, for example, Assumption 3 in \cite{cattaneo2018inference}. Lemma~\ref{lem: assumption check} offers low-level primitive conditions that imply this assumption. The bound $\|\lambda\|_\infty<\infty$ is mild since $\lambda_\ell = \E[x_\ell]$, and for a fixed effects projection matrix, we have $\|M\|_\infty = 2$.\looseness=-1

\section{Inference with clustered data}\label{sec: weak id}

\subsection{Jackknife variance estimator}

Suppose we test the null hypothesis $H_0: \beta = \beta_0$, and define $U_i = y_i - x_i' \beta_0$. Even under $H_0$, these differ from the true errors $e_i$ due to the inclusion of a predictable component: $U_i = e_i + w_i' \delta$. By the partialling-out property of $A^*$, we have $x'A^*U = x'A^*e.$ \looseness=-1

The jackknife variance estimator, introduced by \cite{efron1981jackknife} for symmetric statistics based on independent and identically distributed data, has also been advocated as a useful tool in settings involving non-symmetric statistics and non-identically distributed data \cite[see, e.g.,][]{bousquet2004concentration,mackinnon2023fast}. In our setting, we define the statistic of interest as $\mathcal{Z} = x'A^*U = x'A^*e = f(D_1, \dots, D_N)$, where $D_i = (x_i, U_i)$ denotes the observed data for cluster $i$. Let $\mathcal{Z}_{(i)} = f(D_1, \dots, D_{i-1}, 0, D_{i+1}, \dots, D_N)$ denote the statistic computed when cluster $i$ is removed (i.e., replaced by zero). The jackknife estimator of the variance is then given by
\[
\hat V_\textrm{JK} 
= \sum_{i=1}^N (\mathcal{Z} - \mathcal{Z}_{(i)})^2
= \sum_{i=1}^N \left(\sum_{j \ne i} x_j' A^*_{ji} U_i + \sum_{j=1}^N x_i' A^*_{ij} U_j\right)^2.
\]
\looseness=-1

The following lemma establishes that $\hat V_\textrm{JK}$ is a conservative estimator, potentially overestimating the true variance.
\begin{lemma}\label{lem: jackknife variance}
    Under Assumptions \ref{ass: panel model} and \ref{ass: gaussianity}, the jackknife variance estimator satisfies   $ \E[\hat V_\textrm{JK}] \geq \V(\mathcal{Z}).$    Moreover, if $A^*_{ij} = 0$ for all $i \ne j$ (i.e., $A^*$ is block-diagonal), then the estimator is unbiased under the null: $\E[\hat V_\textrm{JK}] = \V(\mathcal{Z}) = \omega^2$ when $\beta = \beta_0$.
    \looseness=-1
\end{lemma}

\cite{efron1981jackknife} showed that the jackknife variance estimator tends to be conservative for symmetric statistics derived from independent and identically distributed data. While the jackknife correctly captures the variance of first-order $U$-statistics, it double-counts second-order terms. In our setting, $\mathcal{Z}$ is non-symmetric and based on independent but non-identically distributed data. The jackknife is conservative for two main reasons: it gives double weight to the quadratic term (as noted in \cite{efron1981jackknife}), and $\mathcal{Z}_{(i)}$ is improperly centered since $\E[U_i - e_i] \neq 0$. The jackknife variance estimator is unbiased when $A^*_{ij} = 0$ for all $i \ne j$, in which case the removal of a single cluster does not affect the residualization of others. Consequently, $\mathcal{Z}_{(i)}$ is properly centered, and no cross-cluster covariances arise. In general, the degree of conservativeness is small when $\sum_{j=1}^N\sum_{i\neq j}\|A_{ij}^*\|_F^2$ is much smaller than $\sum_{j=1}^N\|A_{jj}^*\|_F^2$.\looseness=-1

\subsection{Tests and confidence sets}

Our estimator, $\hat\beta^{A^*} = \frac{x' A^* y}{x' A^* x} = \frac{z' y}{z' x}$, belongs to the class of just-identified internal IV estimators. Similar to the IV estimator for dynamic panels proposed by \cite{arellano1991some}, our approach may be subject to weak identification concerns (see \cite{bun2010weak}). To construct inference procedures that are robust to weak identification, we propose using the Anderson–Rubin (AR) test. The AR test delivers valid inference under both strong and weak identification.

To test the null hypothesis $H_0: \beta = \beta_0$, we employ the test statistic
\[
AR(\beta_0) = \frac{(x' A^* (y - x \beta_0))^2}{\hat V(\beta_0)},
\]
and reject the null if this statistic exceeds the $1 - \alpha$ quantile of the $\chi^2_1$ distribution. The variance estimator $\hat V(\beta_0)$ is given by $\hat V_\textrm{JK}$, which is based on the implied errors $U = U(\beta_0) = y - x \beta_0$. By Theorem~\ref{thm: gaussianity} and some technical conditions guaranteeing that the variance estimator converges to its expectation, the AR test has asymptotic size no greater than the nominal level. 
\looseness=-1

A confidence set is obtained by inverting the AR test. Since both the numerator and the denominator are quadratic in $\beta_0$, the inversion reduces to solving a quadratic inequality. The resulting confidence set is always non-empty and always includes the estimator $\hat\beta^{A^*}$. 
\looseness=-1

When identification is strong, the AR-based test and confidence set are asymptotically equivalent to those based on the $t$-statistic, using $\hat\beta^{A^*}$ as the point estimator and $\frac{\hat V(\hat\beta^{A^*})}{(x' A^* x)^2}$ as the squared standard error. Nevertheless, we advocate using the AR test, as the validity of $t$-statistic-based inference is not generally guaranteed.

\section{Empirical Application}\label{sec: empirical application}

Policy interventions with spatial effects present a compelling setting for applying our method. In such contexts, researchers often distinguish between the \emph{direct} effect of treatment---on the treated units themselves---and the \emph{total} effect, which includes indirect impacts arising from spillovers. A prominent example is provided by \cite{egger2022general} (hereafter EHMNW), who study a large-scale fiscal intervention in rural Kenya. In their experiment, one-time cash transfers in total equivalent to approximately 15 percent of local GDP were randomly assigned to 653 villages. We apply our method to this setting to estimate the direct effect of the intervention on treated villages.

The villages (our units of observation) were divided into 84 sub-locations (clusters),  the median sub-location contains $7$ villages.  Treated villages were randomly selected through a two-stage randomization procedure.\footnote{First stage assigns sublocations to low- and high-saturation groups (68 groups), second stage randomizes villages to treatment within groups with $1/3$ and $2/3$ propensity.} Within each treated village, all eligible households received a one-time cash transfer of USD 1,871 (PPP-adjusted). Eligibility was determined through a poverty-based means test. The median village includes $98$ households, approximately one-third of whom met the eligibility criteria. EHMNW examine a range of outcomes, including household consumption, firm behavior, and local prices. Here, we focus exclusively on village-aggregated outcomes for total household consumption. Results for other outcomes are similar and available upon request.

The parameter of interest in our application is the village-level \emph{direct} treatment effect on average household consumption. Although treatment was randomly assigned, a simple OLS regression is unreliable due to anticipated interference from neighboring villages. EHMNW assume that spatial spillovers do not extend beyond a radius of $R$ kilometers, with their preferred specification (selected through BIC criteria) being $R = 2$.
  
This setting fits naturally within our doubly robust framework, which we illustrate through estimation under two specifications. In both cases, we define the exogeneity matrix $\mathcal{E}$ using cutoffs based on pairwise distances between villages. We assume  the exclusion restriction $E[v_{\tilde\ell}({y}_\ell - \beta x_\ell)] = 0$ holds if and only if $\ell=\tilde\ell$ or the distance between villages $\tilde\ell$ and $\ell$ is at least $R$ kilometers. In practice, researchers should choose a distance cutoff below which exogeneity may be questionable. Here, we illustrate the behavior of the $\hat\beta^{A^*}$ estimator using a range of cutoffs from 0 to 3 kilometers. In the data, the average pairwise distance between villages within a sub-location is $2.2$ kilometers. On average, each village has $1.5$, $4.6$, and $6.7$ neighbors within the sub-location and within 1, 2, and 3 kilometers\footnote{Based on GPS coordinates. We dropped nine observations for which we do not have GPS coordinates.}, respectively. Our estimates under different choices of $R$ are shown in Figure~\ref{fig:estimates}.

\begin{figure}[htb!]
    \begin{center}\subfloat[Baseline\label{fig:panelA}]{
            \includegraphics[width=0.4\textwidth]{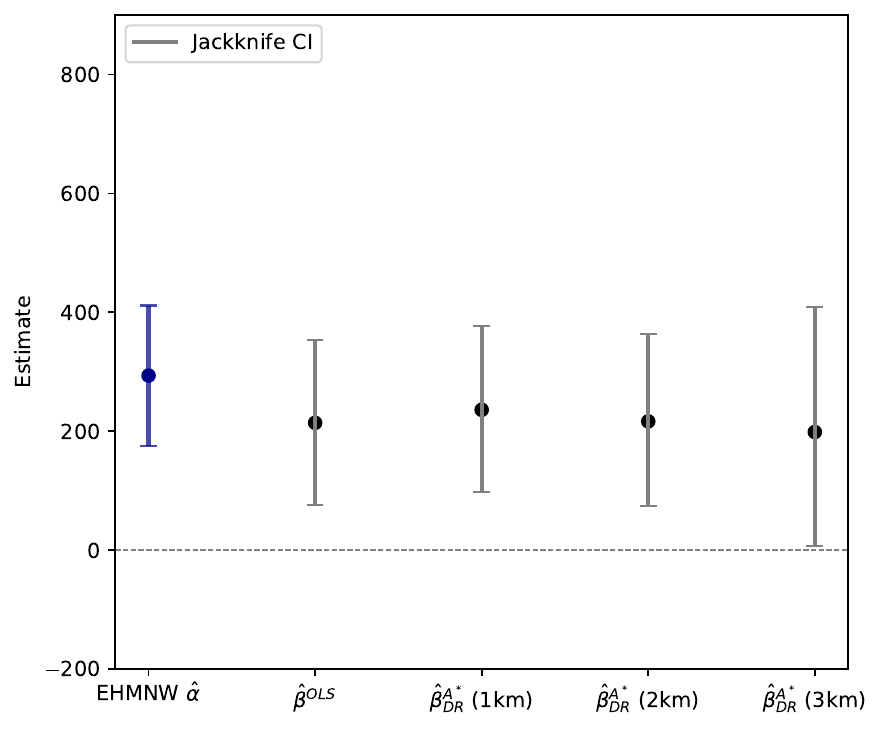}
        }
        \hfill
        \subfloat[Continuous treatment\label{fig:panelB}]{
            \includegraphics[width=0.4\textwidth]{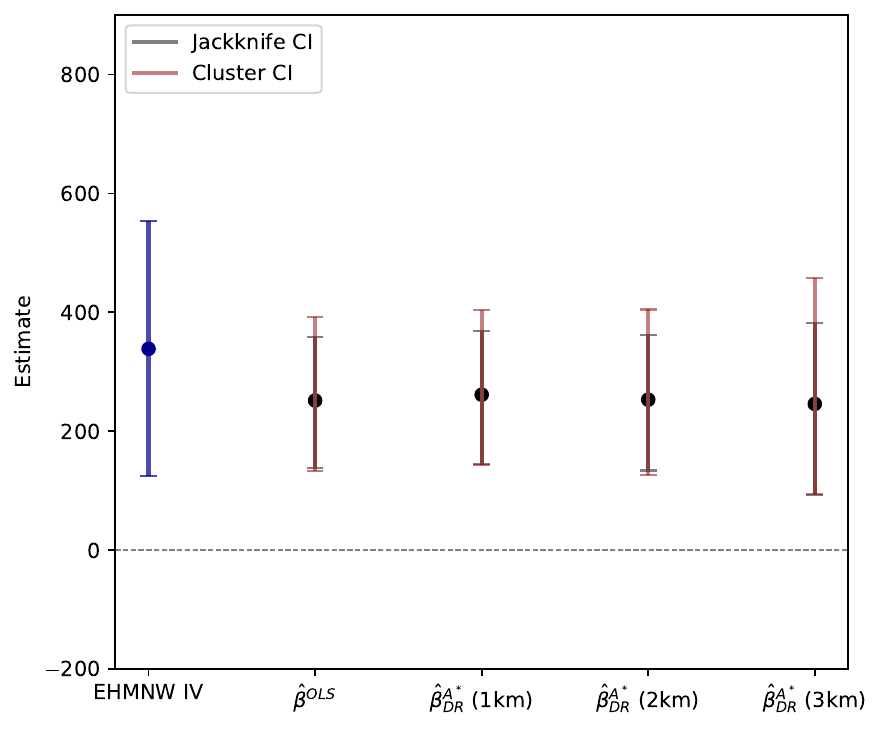}
        }
        \caption{Estimated effects on household consumption}
        \label{fig:estimates}
    \end{center}
    {\footnotesize Notes. The outcome variable in Panel (a) is average household consumption; in Panel (b), it is average household consumption among eligible households. Effect sizes in Panel (b) are scaled by the average transfer amount. Jackknife and cluster-robust confidence intervals are computed by inverting the respective AR test statistics.  The estimates labeled “EHMNW” are taken directly from the original paper: the $\hat{\alpha}$ estimate in Panel (a) is based on Specification (1) from EHMNW and reported in their Table 1 (row 1, column 1); the IV estimate in Panel (b) is based on Specification (2) from the paper and reported in Table 1 (row 1, column 2). These estimates use spatial standard errors based on household GPS coordinates, following Conley (2008).
    }
\end{figure}


For specification (a), let ${y}_\ell$ denote the average consumption of households in village $\ell$, and let $x_\ell$ denote the binary treatment status of the village. We consider the treatment and outcome equations $x_\ell=p_{i(\ell)}+v_\ell$
and $y_\ell=\beta x_\ell+\mu_{i(\ell)}+e_\ell$, 
where $p_i,\mu_i$ are cluster (sub-location) fixed effects. In the treatment equation, these fixed effects flexibly absorb the realized treatment share, which may differ substantially from the ex ante assignment probability because of cluster size, re-randomization, or re-balancing. In the outcome equation, they absorb sub-location-level differences in average consumption. Together with the corresponding exogeneity restrictions encoded by $\mathcal{E}$, our doubly robust estimator remains correctly centered if either the treatment equation or the outcome equation is correctly specified.

The results are reported in Panel (a), alongside the baseline estimator from the specification used in EHMNW. The baseline estimator is reported in Table 1 (row 1, column 1) of EHMNW. We find that the estimates remain relatively stable, suggesting that the bias in estimating the direct effect due to inter-village spillovers is limited. As larger values of $R$ correspond to more relaxed exogeneity assumptions, the effective sample size decreases, leading to larger standard errors.  Figure~\ref{fig:estimates}, Panel (a), also displays confidence intervals based on the jackknife variance estimator described in Section~\ref{sec: weak id}. In this specification, since the only controls are cluster fixed effects, the $A^*$ matrix is block-diagonal, and the jackknife estimator coincides with the standard cluster-robust variance estimator.

Specification (b), reported in Figure~\ref{fig:estimates}, Panel (b), makes use of the full variation in treatment intensity. Here, we define ${y}_\ell$ as the average consumption of eligible households in village $\ell$, and let $x_\ell$ denote the total transfer allocated to the village. The assumed treatment assignment equation in this case is $x_\ell = W_\ell' \delta + v_\ell,$ where $W_\ell$ includes the number of eligible households in village $\ell$ along with sub-location fixed effects. The parameter of interest, $\beta$, again represents the village-level direct treatment effect. This parameter is somewhat comparable to the total effect on eligible households reported in EHMNW, which was estimated using an instrumental variables strategy and reported in Table 1 (row 1, column 2) of the paper. Estimates under different exogeneity assumptions are presented along with both jackknife and cluster-robust standard errors. In this specification, the controls $W_\ell$ include a covariate that varies within clusters and is not absorbed by the sub-location fixed effects. As a result,  matrix $M$ is no longer block-diagonal, and neither is $A^*$, so the two variance estimators do not coincide.

\textbf{Discussion.} There are two main takeaways. First, it is important that researchers state explicitly which exogeneity restrictions they are willing to
  maintain. Our framework then carries these restrictions through estimation and inference in a
  single coherent procedure. Second, there is a tradeoff between the precision of estimates and the strength of the maintained exogeneity assumptions. For example, allowing spillovers out to 3~km rather than 2~km widens the confidence intervals by
  reducing the effective sample size, as measured by $\tr(A^*)$.

\begin{figure}[htb!]
    \begin{center}
        \subfloat[$A^*$ with 1 km cutoff]{
            \includegraphics[width=0.4\textwidth]{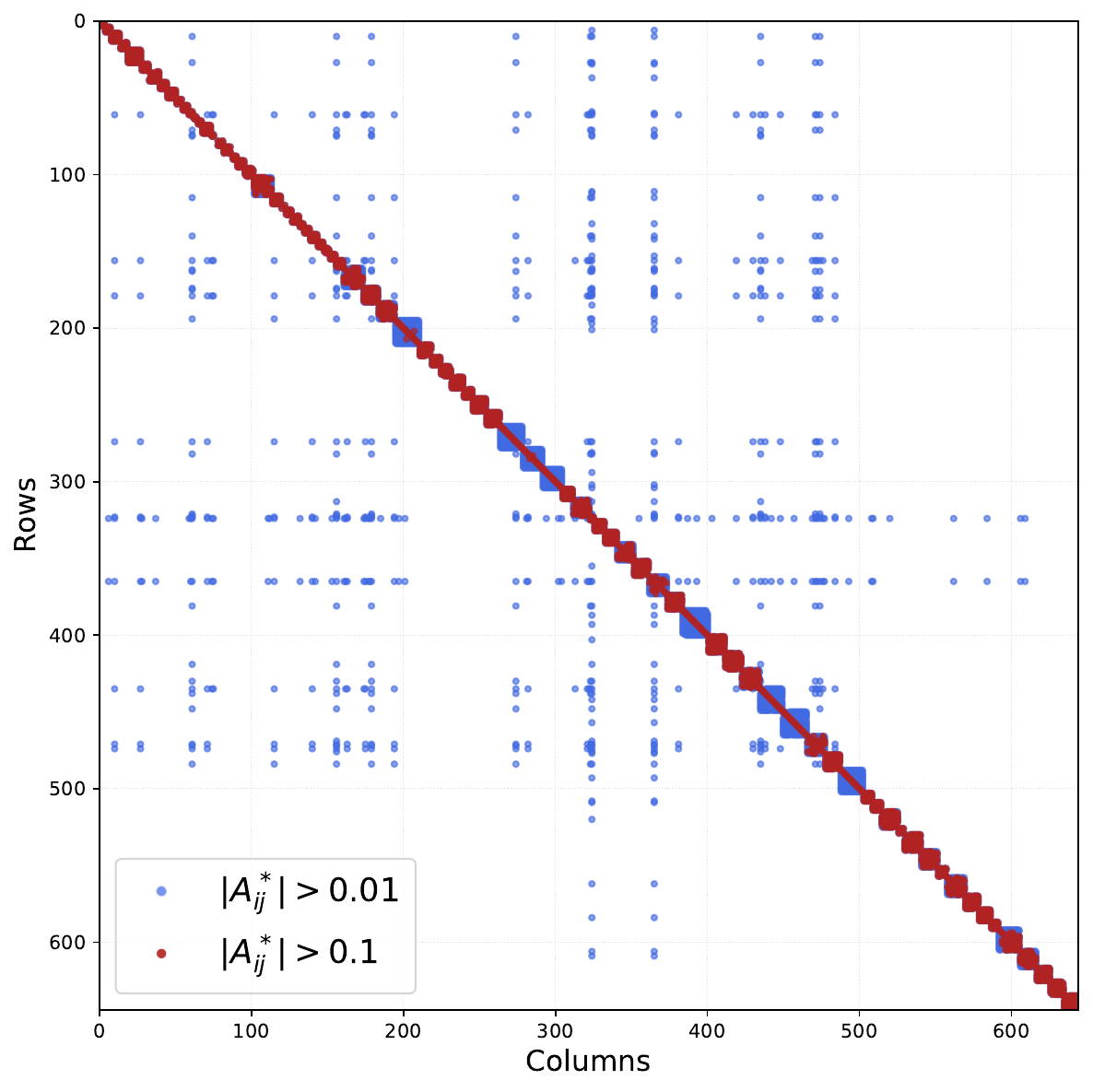}
        }
        \hfill
        \subfloat[$A^*$ with 3 km cutoff]{
            \includegraphics[width=0.4\textwidth]{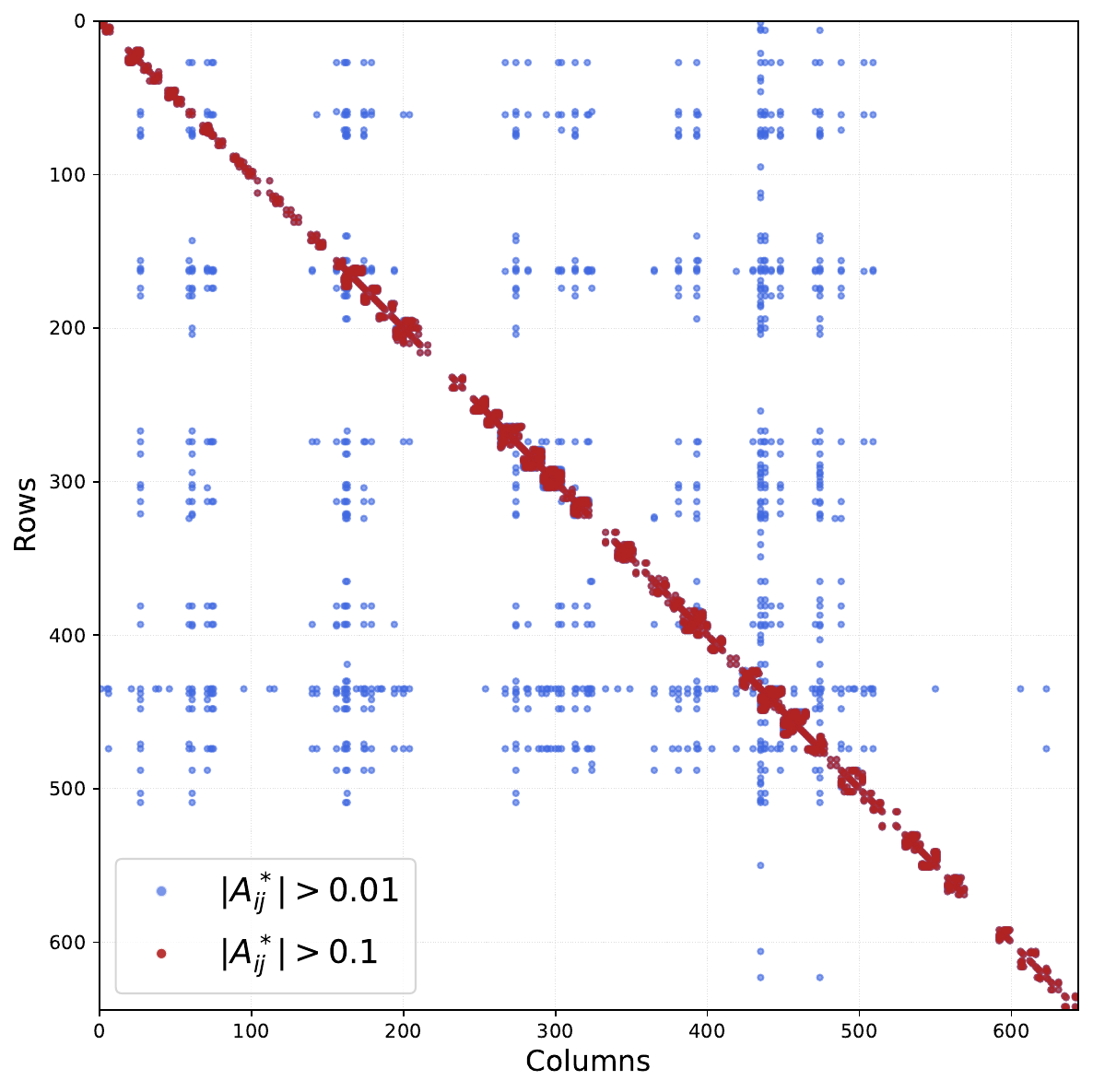}
        }
        \caption{Non-zero structure of $A^*$}\label{fig: A matrix}
    \end{center}
    {\footnotesize Notes. This figure visualizes entries of the matrix $A^*$ with absolute value greater than 0.01 (blue) or 0.1 (red). Panel (a) shows $A^*$ for 1km cutoff and Panel (b) shows $A^*$ for 3km cutoff. 
    }
\end{figure}

The structure of the matrix $A^*$ in the continuous treatment specification (same as in Figure~\ref{fig:estimates}, Panel (b)) is shown in Figure~\ref{fig: A matrix}, Panel (a) for $R = 1$ and Panel (b) for $R = 3$. A blue dot indicates that the absolute value of the corresponding element of $A^*$ exceeds 0.01, while a red dot indicates that it exceeds 0.1. Villages are sorted by sub-location, so that block structure, if present, is visible. The figure reveals that the matrix $A^*$ is far from block-diagonal: some villages contribute substantially to the residualization of controls across many clusters. As the cutoff radius $R$ increases, the trace of the matrix $A^*$ decreases. This pattern reflects a reduction in the effective sample size. The ratio of Frobenius norms for off-diagonal blocks to diagonal blocks remains small. When $R=1$, we have $\tr(A^*)=535$, while the relative Frobenius norm of off-diagonal blocks is $0.044$. When $R=3$, $\tr(A^*)=252$ while the Frobenius norm of off-diagonal blocks is $0.059$. Due to the relative smallness of off-diagonal blocks, the confidence sets using clustered and jackknife variance estimators are close to each other.

\bibliographystyle{chicago}
\bibliography{literature.bib}

\appendix

\section{Appendix }
\begin{ass}\label{ass: technical}
\begin{enumerate}[label=(\roman*)]
        
     \item\label{ass: technical1}    $\max_\ell \E[x_\ell^4 + e_\ell^4] < C$.
    
    \item\label{ass: technical2} The expected variation of the residualized regressor $\tilde x_{\ell}$ diverges with the sample size $n$; $\frac{1}{n}\sum_{\ell=1}^n \E[ \tilde x_{\ell}^2 ] \to Q > 0$.

    \item\label{ass: technical3} The number of controls $K$ may grow with $n$, subject to $\frac{K}{n} < 1 - c$. The largest cluster size satisfies $\max_i T_i / n = o(1)$.

\end{enumerate}

\end{ass}

\begin{proof}[Proof of Lemma \ref{lem: nickell bias}]
    By definition, \( \tilde{x}_\ell = \sum_{\ell^* = 1}^n M_{\ell \ell^*} x_{\ell^*} \). Assumption~\ref{ass: panel model} implies:
\(
\E[ \tilde{x}_\ell e_\ell ] = \sum_{\ell^* =1} ^n M_{\ell \ell^*} \E[ x_{\ell^*} e_\ell ].
\)
Therefore,
\begin{align}
    \frac{1}{n} \sum_{\ell = 1}^n \E[ \tilde{x}_\ell e_\ell ] 
    = \frac{1}{n} \sum_{\ell = 1}^n \sum_{\tilde{\ell}=1}^n M_{\ell \tilde{\ell}} \E[ x_{\tilde{\ell}} e_\ell ].
    \label{eq: bias}
\end{align}

To establish convergence in probability, we use an Efron–Stein-type argument. For generic clustered variables \( (v_\ell, u_\ell) \), with \( \max_\ell \E[v_\ell^4 + u_\ell^4] < C \), define:
\[
\tilde{v}_\ell = \sum_{\ell^* = 1}^n M_{\ell \ell^*} v_{\ell^*}, \quad 
\mathcal{V}_i = \sum_{\ell \in S_i} \tilde{v}_\ell u_\ell, \quad 
\mathcal{U}_i = \sum_{\ell, \ell^* \in S_i} M_{\ell \ell^*} v_\ell u_{\ell^*}.
\]
Note:
$\sum_{\ell = 1}^n \tilde{v}_\ell^2 \le \sum_{\ell = 1}^n v_\ell^2, 
\sum_{\ell \in S_i} \left( \sum_{\ell^* \in S_i} M_{\ell \ell^*} v_{\ell^*} \right)^2 \le \sum_{\ell \in S_i} v_\ell^2.$
Define \( \mathcal{M} = \max_i \sum_{\ell \in S_i} u_\ell^2 \). Then:
\begin{align}
\frac{1}{n^2} \sum_{i=1}^N \E[ \mathcal{V}_i^2 ]
&\le \frac{1}{n^2} \sum_{i=1}^N \E\left[ \sum_{\ell \in S_i} \tilde{v}_\ell^2 \mathcal{M} \right]
\le \E\left[ \frac{\mathcal{M}}{n^2} \sum_{\ell = 1}^n v_\ell^2 \right]
\le C \sqrt{ \E \left( \frac{\mathcal{M}^2}{n^2} \right) } \nonumber \\
&\le C \sqrt{ \E \left[ \frac{1}{n^2} \sum_{i=1}^N \left( \sum_{\ell \in S_i} u_\ell^2 \right)^2 \right] }
\le C \sqrt{ \frac{\max_i T_i}{n} } = o(1). \label{eq: conv1}
\end{align}
Similarly,
\begin{align}
\frac{1}{n^2} \sum_{i=1}^N \E[ \mathcal{U}_i^2 ]
\le \frac{1}{n^2} \sum_{i=1}^N \E\left[ \sum_{\ell \in S_i} v_\ell^2 \sum_{\ell \in S_i} u_\ell^2 \right]
\le C \frac{1}{n^2} \sum_{i=1}^N T_i^2 
\le C \frac{\max_i T_i}{n} = o(1). \label{eq: conv2}
\end{align} 
Now consider:
\(
\mathcal{S}_n = \frac{1}{n} \sum_{\ell = 1}^n \tilde{x}_\ell e_\ell 
= \frac{1}{n} \sum_{\ell = 1}^n x_\ell \tilde{e}_\ell 
= \frac{1}{n} \sum_{\ell = 1}^n \sum_{\ell^* = 1}^n M_{\ell \ell^*} x_\ell e_{\ell^*}.
\)
Define:
\[
\Delta_i \mathcal{S}_n := \mathcal{S}_n - \frac{1}{n} \sum_{\ell \notin S_i} \sum_{\ell^* \notin S_i} M_{\ell \ell^*} x_\ell e_{\ell^*}
= \frac{1}{n} \sum_{\ell \in S_i} \left( \tilde{x}_\ell e_\ell + x_\ell \tilde{e}_\ell \right)
- \frac{1}{n} \sum_{\ell, \ell^* \in S_i} M_{\ell \ell^*} x_\ell e_{\ell^*}.\]
Applying \eqref{eq: conv1} and \eqref{eq: conv2} with \( (v_\ell, u_\ell) \in \{(x_\ell, e_\ell), (e_\ell, x_\ell)\} \), we conclude:\looseness=-1
$\sum_{i=1}^N \E[(\Delta_i \mathcal{S}_n)^2] = o(1).$
By the Efron-Stein inequality, \( \V(\mathcal{S}_n) \le \sum_{i=1}^N \E[ (\Delta_i \mathcal{S}_n)^2 ] \), this implies:
$\frac{1}{n} \sum_{\ell = 1}^n \tilde{x}_\ell e_\ell - \E\left[ \frac{1}{n} \sum_{\ell = 1}^n \tilde{x}_\ell e_\ell \right] = o_p(1).$ 
A similar argument for \( \mathcal{S}_n = \frac{1}{n} \sum_{\ell = 1}^n \tilde{x}_\ell^2 = \frac{1}{n} \sum_{\ell = 1}^n \tilde{x}_\ell x_\ell \), where:
\[
\Delta_i \mathcal{S}_n = \frac{2}{n} \sum_{\ell \in S_i} \tilde{x}_\ell x_\ell - \frac{1}{n} \sum_{\ell, \ell^* \in S_i} M_{\ell \ell^*} x_\ell x_{\ell^*},
\]
implies:
$\frac{1}{n} \sum_{\ell = 1}^n \tilde{x}_\ell^2 - \E\left[ \frac{1}{n} \sum_{\ell = 1}^n \tilde{x}_\ell^2 \right] = o_p(1).$ \looseness=-1 By Assumption~\ref{ass: technical}\ref{ass: technical2}, we have:
$\frac{1}{n} \sum_{\ell = 1}^n \tilde{x}_\ell^2 = Q + o_p(1),$ with  $Q > 0.$ Therefore, combining with \eqref{eq: bias} and applying the continuous mapping theorem, the result follows.    
\end{proof}

\begin{proof}[Proof of Lemma \ref{lem- no unbiased est}.] The following example borrows from Chapter~2.1 of \cite{lehmann1998theory}. Consider i.i.d.\ Bernoulli variables $\{z_i\}_{i=1}^n$ with $p \in (0,1)$ and i.i.d.\ Rademacher variables $\{q_i\}_{i=1}^n$, independent of $\{z_i\}$. Define
$x_i = q_i(\sqrt{2}\, z_i + 1 - z_i), 
\quad 
y_i = q_i(z_i/\sqrt{2} + 1 - z_i), 
\quad i=1,\dots,n.$ Then $x_i y_i = 1$ a.s., $x_i \in \{\sqrt{2}, -\sqrt{2}, 1, -1\}$ with probabilities $p/2, p/2, (1-p)/2, (1-p)/2$, and $\E[x_i] = \E[y_i] = 0$. The data satisfy the correctly specified regression
\[
y_i = \beta x_i + e_i, \quad 
\beta = \frac{\E[y_i x_i]}{\E[x_i^2]} = \frac{1}{1+p}, \quad \E[e_i] = \E[x_i e_i] = 0.
\]
Thus, this is one of the distributions from class $\mathcal{F}$.
Any estimator $u(x,y)$ has expectation
\[
\E[u(x,y)] = \sum_{(x_1,...,x_n,y_1,...,y_n)} u(x_1,...,x_n,y_1,...,y_n)\, p^k (1-p)^{n-k} 2^{-n},
\]
where $k$ is the number of indices with $z_i=1$. Hence $\E[u]$ is a polynomial in $p$ of degree at most $n$ and cannot equal $\beta = (1+p)^{-1}$ for all $p \in (0,1)$. Thus, no unbiased estimator exists for this much simpler model.
\end{proof}

\begin{proof}[Proof of Lemma \ref{lem: OLS not cc}] A possible data generating process has $\E[x_{\tilde \ell} e_\ell] = M_{\tilde\ell \ell} (1-\mathcal{E}_{\tilde \ell  \ell})$. For the outcome model, the lemma then follows from \eqref{eq: bias}. The derivation for the design-based model is similar.
\end{proof}

\begin{proof}[Proof of Lemma \ref{lem: correct centered class}.]
Fix a column coordinate $j$.  Take $F\in\mathcal F$ such that $y=x\beta + W\delta_j + e$ where $\delta_j$ is the unit vector with 1 in the $j$-th position. In addition, let the marginal distribution of $x$ be a point mass at $x_0$, and let $e$ be independent of $x$. In this case, for $u(x,y)$ to be correctly centered on $F$, we have \begin{equation}
    \mathbb E_F[x'Ay] - \beta \mathbb E_F[x'Ax] =\mathbb E_F[x'AW\delta_j]= x_0'AW_j=0.
\end{equation} This holds for all values of $x_0$. This implies $A W_j=0$ for all $j$, hence $AW=0$, i.e.,  $AM=A$.

Finally, take a pair of indexes $(\tilde\ell_0,\ell_0)$ such that $\mathcal E_{\tilde \ell_0,\ell_0}=0$. Take $F\in \mathcal F$ such that $\delta=0,\beta=\beta_0$ and $\mathbb E_F[x_{\tilde\ell}e_{\ell}]=0$ for all $(\tilde\ell,\ell)\neq (\tilde\ell_0,\ell_0)$ and $\E_F[x_{\tilde\ell_0}e_{\ell_0}]=1$. Let $u(x,y) = \frac{x'Ay}{x'Ax}$ be correctly centered, then 
\begin{equation}
    0 = \mathbb E_{F}[x'Ay] - \beta_0 \mathbb E_{F}[x'Ax] = \mathbb E_F[x'Ae] = \sum_{\ell,\tilde\ell}A_{\tilde\ell \ell} E_F[x_{\tilde\ell}e_{\ell}] = A_{\tilde\ell_0 \ell_0}\mathbb E_F[x_{\tilde\ell_0}e_{\ell_0}] .
\end{equation}
Hence we must have $A_{\tilde\ell \ell} = 0 $ whenever $\mathcal E_{\tilde\ell\ell}=0$. For the converse, simply note that (POP) makes \(x'AW\delta=0\), while (CC) and the imposed moment restrictions make \(\mathbb E[x'Ae]=0\).
\end{proof}

\begin{proof}[Proof of Lemma \ref{lem: efficiency statement}.]  
Consider the optimization problem \eqref{eq: alternative optimization}, and suppose the minimum is attained at \( z^* = \tilde A' x \). For any other \( z = A' x \), with \( A \in \mathcal{A} \), define a linear combination \( z_a = z^* + a z \), which also lies in the same class. Minimizing \( \E\|z_a - x\|^2 \) over \( a \), the first-order condition at \( a = 0 \) yields:
$\E[z' z^*] = \E[z' x].$ Now evaluate the objective function:
\[
V(A) = \frac{\sigma^2 \E[z' z]}{[\E(z' x)]^2} = \frac{\sigma^2 \E[z' z]}{[\E(z' z^*)]^2} \geq \frac{\sigma^2}{\E[z^{*\prime} z^*]} = V(\tilde A),
\]
where the inequality follows from the Cauchy--Schwarz inequality and the first-order condition \( \E[z^{*\prime} z^*] = \E[z^{*\prime} x] \). Finally, observe that the objective in \eqref{eq: alternative optimization} simplifies as:
\[
\E\|z - x\|^2 = \tr\left[ (A' - I_n) \E[xx'] (A - I_n) \right] = \lambda \|A - I_n\|_F^2,
\]
which matches the objective in \eqref{eq: optimization problem}. Hence, \( \tilde A = A^* \) minimizes both objectives.
\end{proof}

\begin{proof}[Proof of Theorem \ref{thm: leave out}.]
\textbf{Part \ref{thm: leave out1}.} Consider the optimization problem:
\begin{equation}
    \|A - M\|_F^2 = (\mathrm{vec}(A) - \mathrm{vec}(M))' (\mathrm{vec}(A) - \mathrm{vec}(M)) \to \min \quad \text{s.t. } \mathcal{L} \, \mathrm{vec}(A) = 0.
\end{equation}
The solution is given by:
$\mathrm{vec}(A^*) = \left(I_{n^2} - \mathcal{L}' (\mathcal{L} \mathcal{L}')^+ \mathcal{L} \right) \mathrm{vec}(M).$

\textbf{Part \ref{thm: leave out2}.} Since $\mathcal{A}$ is a linear subspace and the Frobenius norm arises from an inner product, the solution to the projection problem
\[
A^* = \arg\min_{A \in \mathcal{A}} \|A - M\|_F^2 = \mathrm{proj}_F(M, \mathcal{A})
\]
is the orthogonal projection of $M$ onto $\mathcal{A}$. Hence, $M - A^*$ is orthogonal to any $A \in \mathcal{A}$ in Frobenius inner product: \( \langle M - A^*, A \rangle_F = 0 \).

Consider a matrix $B$ with property $B_{\tilde\ell \ell} = 0$ whenever $\mathcal{E}_{\tilde\ell\ell}=1$, then
\begin{equation}\label{eq: orthogonality property}
    \langle A,BM\rangle_F=\tr(A'BM)=\tr(MA'B)=\tr((AM)'B)=\tr(A'B)=\sum_{\ell_1,\ell_2}A_{\ell_1,\ell_2}B_{\ell_1,\ell_2}=0,
\end{equation}
since for any pair $(\ell_1, \ell_2)$, either $B_{\ell_1,\ell_2} = 0$ or $A_{\ell_1,\ell_2} = 0$ for all $A \in \mathcal{A}$.
then $\langle BM, A \rangle_F = 0$ for all $A \in \mathcal{A}$. If there exists a matrix $B$ with that property and such that  $A^* = M - BM \in \mathcal{A}$, then $A^*$ must be the projection of $M$ onto $\mathcal{A}$. To find such $B$, we need $(M - BM)_{\tilde\ell \ell} = 0$  whenever $\mathcal{E}_{\tilde\ell\ell}=0$. This condition leads to the linear system:
\begin{equation}
  M_{\tilde\ell \ell} - \sum_{\ell_1 : \mathcal{E}_{\tilde\ell\ell_1}=0} B_{\tilde\ell \ell_1} M_{\ell_1 \ell} = 0, \quad \text{ whenever }\mathcal{E}_{\tilde\ell\ell}=0.  
\end{equation}
This system has as many equations as unknowns. Define $M_{\tilde\ell}$ to be sub-matrix of $M$ containing only elements $M_{\ell_1\ell_2}$ for indexes such that $\mathcal{E}_{\tilde\ell\ell_1}=0$ and $\mathcal{E}_{\tilde\ell\ell_2}=0$. If the matrix $M_{\tilde\ell}$ is full rank, the solution is:
\begin{equation}
    B_{\tilde\ell \ell_1} = \sum_{\ell:  \mathcal{E}_{\tilde\ell\ell}=0} M_{\tilde\ell \ell} \left(M_{\tilde\ell}^{-1} \right)_{\ell \ell_1},
\quad \text{for } \ell_1 \text{ s.t. } \mathcal{E}_{\tilde\ell\ell_1}=0, \text{ and } 0 \text{ otherwise}.
\end{equation}
If $M_{\tilde\ell}$ is rank-deficient, the system still has a solution if having
$\sum_{\ell: \mathcal{E}_{\tilde\ell\ell}=0}  c_\ell M_{\ell_1, \ell} = 0 $ for all $\ell_1 $ such that $ \mathcal{E}_{\tilde\ell\ell_1}=0$ implies that 
$\sum_{\ell: \mathcal{E}_{\tilde\ell\ell}=0} c_\ell M_{\tilde\ell \ell} = 0.$ This holds since $M_{\ell_1, \ell} = -\tilde w_{\ell_1}' \tilde w_\ell$, where $\tilde w_\ell = (W'W)^{-1/2} w_\ell$, so:
\[
\sum_{\ell: \mathcal{E}_{\tilde\ell\ell}=0} c_\ell M_{\ell_1, \ell} = 0 \Rightarrow
\sum_{\ell: \mathcal{E}_{\tilde\ell\ell}=0} c_\ell \tilde w_\ell = 0 \Rightarrow
\sum_{\ell: \mathcal{E}_{\tilde\ell\ell}=0} c_\ell M_{\tilde\ell \ell} = 0.
\]
Thus, the system is solvable, and we can choose
$B_{\tilde\ell \ell_1} = \sum_{\ell: \mathcal{E}_{\tilde\ell\ell}=0} M_{\tilde\ell \ell} \left(M_{\tilde\ell}^{+} \right)_{\ell \ell_1}$
as a concrete solution using a generalized inverse.

\textbf{Part \ref{thm: leave out3}.}  
Let $U = W'W$, and define selection matrices \( \mathbf{i}_{\tilde\ell}=\mathrm{diag}(\mathcal{E}_{\tilde\ell,\cdot}) \) and \( \mathbf{i}_{\tilde\ell^c}=I_n-\mathbf{i}_{\tilde\ell} \) as diagonal matrices selecting observations whose errors are uncorrelated and correlated with $x_{\tilde\ell}$, respectively. Let $W_{\tilde\ell} = \mathbf{i}_{\tilde\ell} W, \quad W_{\tilde\ell^c} = \mathbf{i}_{\tilde\ell^c} W,$ so that:
\begin{equation}
    U = W'W = W_{\tilde\ell}' W_{\tilde\ell} + W_{\tilde\ell^c}' W_{\tilde\ell^c}.
\end{equation}
By the generalized Woodbury identity \citep{henderson1981deriving}:
\begin{equation}
    (W_{\tilde\ell}' W_{\tilde\ell})^+ = U^{-1} + U^{-1} W_{\tilde\ell^c}' (I_n - W_{\tilde\ell^c} U^{-1} W_{\tilde\ell^c}')^+ W_{\tilde\ell^c} U^{-1}.
\end{equation}
Using notation: $M = I_n - W U^{-1} W', \quad M^{(\tilde\ell)} = I_n - W_{\tilde\ell} (W_{\tilde\ell}' W_{\tilde\ell})^+ W_{\tilde\ell}'.$
The term
\[
\mathbf{i}_{\tilde\ell^c}' (I_n - W_{\tilde\ell^c} U^{-1} W_{\tilde\ell^c}')^+ \mathbf{i}_{\tilde\ell^c} = \mathbf{i}_{\tilde\ell^c} ( \mathbf{i}_{\tilde\ell^c} M \mathbf{i}_{\tilde\ell^c})^+ \mathbf{i}_{\tilde\ell^c}
\]
is a submatrix of $M$. Pre- and post-multiplying the generalized Woodbury formula by the $\tilde\ell$-th and $\tilde\ell_1$-th rows of $W$, we obtain:
\begin{equation}
    M^{(\tilde\ell)}_{\tilde\ell \tilde\ell_1} = M_{\tilde\ell \tilde\ell_1} - \sum_{\ell_1, \ell_2:  \mathcal{E}_{\tilde\ell\ell_1}=0, \mathcal{E}_{\tilde\ell\ell_2}=0} M_{\tilde\ell \ell_1} (M_{\tilde\ell}^{+})_{\ell_1, \ell_2} M_{\ell_2, \tilde\ell_1}
= M_{\tilde\ell \tilde\ell_1} - \sum_{\ell_2: \mathcal{E}_{\tilde\ell\ell_2}=0 } B_{\tilde\ell \ell_2} M_{\ell_2, \tilde\ell_1} = A^*_{\tilde\ell \tilde\ell_1},
\end{equation}
where the last equality follows from part \ref{thm: leave out2}.
\end{proof}

\begin{lemma}\label{lem: helpful inequalities}
Assume a random vector $\xi \in \mathbb{R}^T$ satisfies $\E[\xi] = 0$, $\V(\xi) = I_T$, and for all non-random vectors $\tau$, $\E|\tau' \xi|^4 \leq C \|\tau\|^4$. Then for any positive semidefinite matrix $A \in \mathbb{R}^{T \times T}$,
\[
\E[(\xi' A \xi)^2] \leq C (\tr(A))^2.
\]
\end{lemma}
\begin{proof}[Proof of Lemma \ref{lem: helpful inequalities}.]
As $A$ is positive semidefinite, it has spectral decomposition \( A = \sum_{j=1}^T \mu_j \psi_j \psi_j' \), where $\mu_j \ge 0$ and $\{\psi_j\}$ is an orthonormal set of eigenvectors. Then:
\begin{align*}
\E[(\xi' A \xi)^2] &= \E\left( \sum_{j=1}^T \mu_j (\psi_j' \xi)^2 \right)^2 = \sum_{j,k=1}^T \mu_j \mu_k \E[(\psi_j' \xi)^2 (\psi_k' \xi)^2] \\
&\leq \sum_{j,k=1}^T \mu_j \mu_k \sqrt{\E(\psi_j' \xi)^4} \sqrt{\E(\psi_k' \xi)^4} \leq C \sum_{j,k=1}^T \mu_j \mu_k = C (\tr(A))^2.\qedhere
\end{align*}
\end{proof}

We prove Lemma~\ref{lem: CLT for quadratics} and Theorem~\ref{thm: gaussianity} using a version of \cite{solvsten2020robust}, Lemma~A2.1, that better fits the current setup.

\begin{rem}[Forward--reverse averaging in \cite{solvsten2020robust}, Lemma~A2.1]
\label{rem:forward-reverse}
For a statistic $W = W(\xi)$, constructed from $\xi=(\xi_1,\ldots,\xi_N)'$ with independent coordinates, let $\tilde \xi = (\tilde \xi_1,\ldots,\tilde \xi_N)'$ be an independent copy of $\xi$. Condition (i) of \cite{solvsten2020robust}, Lemma~A2.1, requires one to verify that $\E[T_N^F\mid \xi]\xrightarrow{L^1}1$ where $T_N^F = \frac12\sum_{i=1}^N\Delta_iW\,\Delta_i^FW$ for $\Delta_iW = W(\xi) - W(\xi_1,\dots,\xi_{i-1},\tilde \xi_i,\xi_{i+1},\dots,\xi_N)$, $\Delta_i^F W = W(\tilde \xi_1,\dots,\tilde \xi_{i-1}, \xi_i,\dots,\xi_N) - W(\tilde \xi_1,\dots,\tilde \xi_i,\xi_{i+1},\dots,\xi_N)$. The lemmas conclusion continues to hold under its remaining conditions if (i) is verified for $\bar T_N=\frac12(T_N^F+ T_N^R)$ instead of $T_N^F$ where $T_N^R = \frac12\sum_{i=1}^N\Delta_iW\,\Delta_i^RW$ and $\Delta_i^RW = W( \xi_1,\dots,  \xi_i,\tilde \xi_{i+1},\dots,\tilde \xi_N) - W( \xi_1,\dots, \xi_{i-1},\tilde \xi_i,\dots,\tilde \xi_N)$.
Indeed, the covariance identity in \cite{solvsten2020robust}, Lemma~A2.10, holds for both the forward and reverse telescoping schemes. Averaging the two identities replaces the ordering-specific $T_N^F$ by $\bar T_N$, while the remainder bound in the proof of \cite{solvsten2020robust}, Lemma~A2.1, is unchanged because $\Delta_i^F W$, $\Delta_i^R W$, and $\Delta_iW$ have the same marginal distributions.\looseness=-1
\end{rem}

\begin{proof}[Proof of Lemma \ref{lem: CLT for quadratics}.] For $W=\frac{1}{\sqrt{V}} \sum_i\sum_{j \ne i}\xi_i'  \Gamma_{ij} \xi_j$
define $\Delta^0_i W = \frac{2}{\sqrt{V}} \xi_i' \sum_{j \ne i} \Gamma_{ij} \xi_j.$ To apply \cite{solvsten2020robust}, Lemma A2.1, we must first verify:
\[
\sum_{i=1}^N \E[(\Delta_i^0 W)^2] = O(1) \quad \text{and} \quad \sum_{i=1}^N \E[(\Delta_i^0 W)^4] \to 0.
\]
For the second moment:
$\sum_i \E[(\Delta_i^0 W)^2] = \frac{4}{V} \sum_{i=1}^N \sum_{j \ne i} \tr(\Gamma_{ij} \Gamma_{ij}') = 2.$
For the fourth moment, decompose the expression into two terms,
\begin{align*}
   \sum_i\E[|\Delta_i^0 W|^4] \le \frac{C}{V^2}\sum_{i=1}^N\E\left[\sum_{j\neq i}(\xi_i'\Gamma_{ij}\xi_j)^4+\sum_{j\neq i}\sum_{k\notin \{i,j\}}(\xi_i'\Gamma_{ij}\xi_j)^2(\xi_i'\Gamma_{ik}\xi_k)^2\right].
\end{align*}
The first term involves:
\begin{align*}
      &\frac{1}{V^2}\E\left\{\sum_{i=1}^N\sum_{j\neq i}\E\left[(\xi_i'\Gamma_{ij}\xi_j)^4|\xi_j\right]\right\}\leq \frac{C}{V^2}\E\sum_i\sum_{j\neq i}\|\Gamma_{ij}\xi_j\|^4 \\ 
      &= \frac{C}{V^2}\E\left\{\sum_i\sum_{j\neq i}\left(\xi_j'\Gamma_{ij}'\Gamma_{ij}\xi_j\right)^2\right\}\leq  \frac{C}{V^2}\sum_i\sum_{j\neq i}\left(\tr(\Gamma_{ij}'\Gamma_{ij})\right)^2= \frac{C}{V^2}\sum_i\sum_{j\neq i}\|\Gamma_{ij}\|_F^4,
\end{align*}
where we used Lemma \ref{lem: helpful inequalities}.
For the second term, we notice that
\begin{align*}
  &\E[ (\xi_i'\Gamma_{ij}\xi_j)^2(\xi_i'\Gamma_{ik}\xi_k)^2]=\E(\xi_i'\Gamma_{ij}\Gamma_{ij}'\xi_i\xi_i'\Gamma_{ik}\Gamma_{ik}'\xi_i) \\
   &\leq \sqrt{\E[(\xi_i'\Gamma_{ij}\Gamma_{ij}'\xi_i)^2]}\sqrt{\E[(\xi_i'\Gamma_{ik}\Gamma_{ik}'\xi_i)^2]}\leq C\|\Gamma_{ij}\|_F^2\|\Gamma_{ik}\|^2_F.
\end{align*}
Combining both yields
\begin{align*}
    \sum_i\E|\Delta_i^0 W|^4\leq\frac{C}{V^2}\sum_i\sum_{j\neq i}\left(\|\Gamma_{ij}\|_F^4+\sum_{k\notin \{i,j\}}\|\Gamma_{ij}\|_F^2\|\Gamma_{ik}\|^2_F\right)\leq\frac{C\max_i\left(\sum_{j\neq i}\|\Gamma_{ij}\|_F^2\right)}{V}\to 0.
\end{align*}

We finally need to verify convergence of the conditional variance term (condition (i)):
\[
\frac{1}{V} \sum_{i=1}^N \sum_{k \neq i} \sum_{j \ne i} \xi_j' \Gamma_{ji} (\xi_i \xi_i' + I_{T_i}) \Gamma_{ik} \xi_k \xrightarrow{L^1} 1.
\]
This expression decomposes into five terms: \( I_1 + 2I_2 + I_3 + 3I_4 + \frac{2}{V} \sum_{i=1}^N \sum_{k \neq i}\|\Gamma_{ik}\|_F^2 \), where
\begin{align}
I_1 &= \frac{1}{V} \sum_{i=1}^N \sum_{k\neq i} \sum_{j \notin \{i,k\}} \xi_j' \Gamma_{ji} (\xi_i \xi_i' - I_{T_i}) \Gamma_{ik} \xi_k , \label{eq:st1} \\
I_2 &= \frac{1}{V} \sum_{i=1}^N \sum_{k \neq i} \sum_{j \notin \{i,k\}} \xi_j' \Gamma_{ji} \Gamma_{ik} \xi_k , \label{eq:st4} \\
I_3 &= \frac{1}{V} \sum_{i=1}^N \sum_{k \neq i} \tr \left( \Gamma_{ik} (\xi_k \xi_k' - I_{T_k}) \Gamma_{ki} (\xi_i \xi_i' - I_{T_i}) \right) , \label{eq:st2} \\
I_4 &= \frac{1}{V} \sum_{i=1}^N \sum_{k \neq i} \tr \left( (\xi_k \xi_k' - I_{T_k}) \Gamma_{ki} \Gamma_{ik} \right).\label{eq:st3}
\end{align}
The final term in the decomposition is equal to one as \( V = 2\|\Gamma\|_F^2 \), so it suffices to show that \( I_1, I_2, I_3, I_4 \xrightarrow{L^1} 0 \), which is implied by the $L^2$-convergence established below.

In all cases, the expectation of the left-hand side is zero, so we only need to check that the variance of each expression converges to zero. Terms belonging to different unordered triples in $I_1$ are uncorrelated, so that
\begin{align}
    \V(I_1) &\leq \frac{C}{V^2}
\sum_{i=1}^N\sum_{j\ne i}\sum_{k\notin\{i,j\}}
\mathbb E\text{tr}
\left(\Gamma_{ji}(\xi_i\xi_i^\prime\!-\!I_{T_i})\Gamma_{ik}\Gamma_{ki}(\xi_i\xi_i^\prime\!-\!I_{T_i})\Gamma_{ij}
\right) \\
&\le \frac{C}{V^2}
\sum_{i=1}^N
\mathbb E\operatorname{tr}
\left((\xi_i\xi_i^\prime\!-\!I_{T_i})\bar \Gamma_{ii}(\xi_i\xi_i^\prime\!-\!I_{T_i})\bar \Gamma_{ii}
\right) \\
&\le \frac{C}{V^2}
\sum_{i=1}^N \left( \text{tr}(\bar \Gamma_{ii}) \right)^2 \le 
\frac{C}{V}
\max_i\sum_{j\ne i}\|\Gamma_{ij}\|_F^2
\to 0,
\end{align}
where the last inequality applies Lemma \ref{lem: helpful inequalities} and we defined $\bar \Gamma = \Gamma^2$.
Thus, $I_1\xrightarrow{L^2} 0$. Summands in $I_2$ are uncorrelated unless the set of indices $j,k$ coincides, so
\begin{align*}
    \V(I_2) = \V(\frac{1}{V} \sum_{k=1}^N \sum_{j \neq k} \xi_j' \bar \Gamma_{jk}\xi_k)
    = \frac{2}{V^2}\sum_{k = 1}^N\sum_{j\neq k}\tr(\bar \Gamma_{jk}\bar \Gamma_{kj})
    \le \frac{2\|\Gamma^2\|_F^2 }{V^2}
    \le \frac{C\|\Gamma\|^2\|\Gamma\|^2_F}{\|\Gamma\|^4_F} \to 0.
\end{align*}
Thus, $I_2\xrightarrow{L^2} 0$. Introduce $\eta_{ik}=\tr\left(\Gamma_{ik}(\xi_k\xi_k'-I_{T_k})\Gamma_{ki}(\xi_i\xi_i^\prime-I_{T_i})\right)=(\xi_i'\Gamma_{ik}\xi_k)^2-\xi_i'\Gamma_{ik}\Gamma_{ki}\xi_i-\xi_k'\Gamma_{ki}\Gamma_{ik}\xi_k+\|\Gamma_{ik}\|_F^2$ and note by repeated use of Lemma~\ref{lem: helpful inequalities} that the second moments of all these four terms are bounded by $C\|\Gamma_{ik}\|_F^4$.
As $\eta_{ik}$ is uncorrelated with any other $\eta$ except $\eta_{ki}$, we therefore have
\begin{align}
\V(I_3)=\V\left(\frac{1}{V}\sum_{i=1}^N\sum_{k\neq i}\eta_{ik}\right) =\frac{2}{V^2}\sum_{i=1}^N\sum_{k\neq i}\E[\eta_{ik}^2] \leq \frac{C\max_{ik}\|\Gamma_{ik}\|_F^2}{V}\to 0.
\end{align}
Finally, to prove $I_4\xrightarrow{L^2}0$, we note that all grouped summands are independent, so that
\begin{align*}
    &\V\left(\frac{1}{V}\sum_{k=1}^N\xi_k'\bar \Gamma_{kk}\xi_k\right)=\frac{1}{V^2}\sum_{k=1}^N\V\left(\xi_k'\bar \Gamma_{kk}\xi_k\right)\leq \frac{1}{V^2}\sum_{k=1}^N\E\left(\xi_k'\bar\Gamma_{kk}\xi_k\right)^2 \\ &\leq \frac{C}{V^2}\sum_{k=1}^N(\tr \bar\Gamma_{kk})^2=\frac{C}{V^2}\sum_{k=1}^N(\tr \sum_{j}\Gamma_{kj}\Gamma_{jk})^2=\frac{C}{V^2}\sum_{k=1}^N( \sum_{j}\|\Gamma_{kj}\|_F^2)^2\to 0. \qedhere
\end{align*}
\end{proof}

\begin{proof}[Proof of Theorem \ref{thm: gaussianity}.] 
Define $\tilde A_{ij}=\Sigma_i^{1/2}\left(\begin{array}{cc}
    0 & 0 \\
    A_{ij} &0 
\end{array}\right)\Sigma_j^{1/2}$ to be  a $(2T_i) \times (2T_j)$ matrix. Construct the full $(2n) \times (2n)$ matrix $\tilde A$ with $\tilde A_{ij}$ in block $(i,j)$ and $\tilde A_{ii} = 0$ for all $i$. Note that $\tilde A$ is generally not symmetric, $\Gamma = \frac{1}{2}(\tilde A + \tilde A')$ is. The variance of the quadratic component is\looseness=-1
    $$\V\left(\sum_j\sum_{i\neq j}v_i'A_{ij}e_j\right)=\V\left(\sum_j\sum_{i\neq j}\xi_i'\tilde A_{ij}\xi_j\right)=\sum_j\sum_{i\neq j}\tr ( \tilde A_{ij} \tilde A_{ij}'+ \tilde A_{ij} \tilde A_{ji})=
 \underbrace{\tr ( \tilde A \tilde A')+\tr ( \tilde A ^2)}_{=2\norm{\Gamma}_F^2}.$$
 The first term, $\tr(\tilde A \tilde A')$, corresponds to $\sum_j \V(Q_j e_j)$ in the decomposition of Equation \eqref{eq: quad_var}. The second term, $\tr(\tilde A^2)$, captures the correction due to cross-cluster dependence. 
 
 If two random variables $\eta_1$ and $\eta_2$ have correlation bounded  by $1 - c$, then $\V(\eta_1 + \eta_2) \geq c(\V(\eta_1) + \V(\eta_2))$. Assumption~\ref{ass: gaussianity}\ref{ass: gaussianity3} on bounded alignment and Assumption~\ref{ass: gaussianity}\ref{ass: gaussianity1} imply
\begin{align}
    &\max_{\tau_1, \tau_2} \left| \mathrm{corr}(\tau_1' e_i, \tau_2' (v_i \otimes e_i)) \right|  \leq 1 - c,\\ 
    \omega^2 \geq \sum_{j=1}^N c &\left( \V((\lambda' A)_j e_j) + \V(\mathrm{vec}(A_{jj}')'  (v_j \otimes e_j)) \right) + 2\norm{\Gamma}_F^2.
\end{align}
Notice that $2\norm{\Gamma}_F^2 \geq C^{-2}\sum_{i\neq j}\|A_{ij}\|_F^2$ due to Assumption~\ref{ass: gaussianity}\ref{ass: gaussianity1} and
\begin{equation}
\label{eq: low bound on secon variance} 
\V(v_j' A_{jj} e_j) = \V(\mathrm{vec}(A_{jj}')' (v_j \otimes e_j)) \geq c \cdot \mathrm{vec}(A_{jj}')' (\Sigma_{v,j} \otimes \Sigma_{e,j}) \mathrm{vec}(A_{jj}') \ge c \|  A_{jj} \|_F^2.
\end{equation}
Combining the above with Assumption~\ref{ass: gaussianity matrix A}\ref{ass: gaussianity matrix A1}, we conclude $\omega^2 \geq c n$.

Next, recall that: $x' A e = \sum_{j=1}^N \omega_j e_j + \sum_{j=1}^N \sum_{i \neq j} v_i' A_{ij} e_j.$ Along subsequences where the quadratic term is asymptotically negligible in variance, we first establish asymptotic normality for the linear component, $\sum_{j=1}^N \omega_j e_j$, using Lyapunov's condition. For $\sum_j (\lambda' A)_j e_j$, the Lyapunov bound follows from the directional fourth-moment bound in Assumption~\ref{ass: gaussianity}\ref{ass: gaussianity4}, the negligibility condition in Assumption~\ref{ass: gaussianity matrix A}\ref{ass: gaussianity matrix A2} and $\omega^2 \geq c n$. For $\sum_j v_j' A_{jj} e_j$, we apply the higher-moment bound in Assumption~\ref{ass: gaussianity}\ref{ass: gaussianity4}:
\begin{align*}
\E |v_j' A_{jj} e_j|^{2 + 2\delta} &\leq \E |\mathrm{vec}(A_{jj}')' (v_j \otimes e_j)|^{2 + 2\delta} 
\leq C \left( \mathrm{vec}(A_{jj}')' \Phi_j \mathrm{vec}(A_{jj}') \right)^{1 + \delta} \\
&\leq C \left( \mathrm{vec}(A_{jj}')' (\Sigma_{v,j} \otimes \Sigma_{e,j}) \mathrm{vec}(A_{jj}') \right)^{1 + \delta} \\ &\le C\left(
\max_j\|\Sigma_j\|^2\|A_{jj}\|_F^2
\right)^\delta
\operatorname{vec}(A_{jj}')'
(\Sigma_{v,j}\otimes\Sigma_{e,j})
\operatorname{vec}(A_{jj}').
\end{align*}
where the second inequality follows from the upper bound in Assumption~\ref{ass: gaussianity}\ref{ass: gaussianity2}. Using  Assumption~\ref{ass: gaussianity matrix A}\ref{ass: gaussianity matrix A2}, we then have
\begin{align*}
\frac{\sum_{j=1}^N\E\left|v_j'A_{jj}e_j\right|^{2+2\delta}}{\omega^{2+2\delta}}
&\leq C\frac{\left(
\max_j\|\Sigma_j\|^2\|A_{jj}\|_F^2
\right)^\delta \sum_{j=1}^N\operatorname{vec}(A_{jj}')'
(\Sigma_{v,j}\otimes\Sigma_{e,j})
\operatorname{vec}(A_{jj}')}{\omega^{2\delta}\cdot\omega^{2}} \\
&\leq C\left( \frac{\max_j\|\Sigma_j\|^2\|A_{jj}\|_F^2}{n} \right)^\delta\to 0.
\end{align*}
Thus, the linear part satisfies
$\frac{1}{\omega} \sum_{j=1}^N \omega_j e_j \xrightarrow{d} \mathcal{N}(0,1).$

Along subsequences where the linear term is asymptotically negligible in variance, we now apply Lemma \ref{lem: CLT for quadratics} to the quadratic part:
$\sum_{i=1}^N \sum_{j \neq i} v_i' A_{ij} e_j = \sum_{i,j} \xi_i' \tilde A_{ij} \xi_j.$ Assumption \ref{ass: gaussianity matrix A}\ref{ass: gaussianity matrix A3}, part (b), ensures that the conditions of Lemma \ref{lem: CLT for quadratics} are satisfied. Thus, the quadratic part satisfies $\frac{1}{\omega} \sum_{i=1}^N \sum_{j \neq i} v_i' A_{ij} e_j \xrightarrow{d} \mathcal{N}(0,1).$

Finally, consider subsequences where neither term is asymptotically negligible in variance. We apply Lemma~A2.1, to $W=\frac{1}{\omega} \sum_{j=1}^N \omega_j e_j +\frac{1}{\omega} \sum_{i=1}^N \sum_{j \neq i} v_i' A_{ij} e_j$. Conditions (ii) and (iii) of \cite{solvsten2020robust}, Lemma~A2.1, follow from the component-wise bounds established above. For its condition (i), the pure-linear and pure-quadratic contributions under the forward--reverse averaging in Remark~\ref{rem:forward-reverse} are, respectively, $\omega^{-2}\sum_i\V(\omega_i e_i)+o_{L^1}(1)$ and $2\omega^{-2}\|\Gamma\|_F^2+o_{L^1}(1)$. Their leading terms sum to one, so it remains only to show that the mixed contribution
$\frac{3}{2\omega^2}\sum_{i\ne j}
\left\{(\omega_i e_i)\xi_i+\E[(\omega_i e_i)\xi_i]\right\}'
\Gamma_{ij}\xi_j$
vanishes in $L^1$. 

Write the expression in braces as its centered part plus twice its expectation. Let $p^{-1}=(2+2\min\{\delta,1\})^{-1}+1/4$ where $p \in (4/3,2]$. The same moment and bounded-alignment bounds used for the linear term, together with H\"older's inequality, imply for every deterministic $\tau$ that
\[
\E\left|\tau'\left\{(\omega_i e_i)\xi_i-\E[(\omega_i e_i)\xi_i]\right\}\right|^p
\leq C\V(\omega_i e_i)^{p/2}\|\tau\|^p.
\]
Splitting the double sum into $j<i$ and $j>i$, each part is a
martingale-difference sum, so
\begin{align*}
\E\left|\frac{1}{\omega^2}\sum\nolimits_{i\ne j}
\left\{(\omega_i e_i)\xi_i-\E[(\omega_i e_i)\xi_i]\right\}'
\Gamma_{ij}\xi_j\right|^p 
\leq C\left(\frac{\max_i\sum_{j\ne i}\|\Gamma_{ij}\|_F^2}{\omega^2}\right)^{(p-1)/2}\to0.
\end{align*}
For the non-centered part, Cauchy--Schwarz gives $\|\E[(\omega_i e_i)\xi_i]\|^2\leq\V(\omega_i e_i)$, and hence
\[
\E\left|\frac{1}{\omega^2}\sum\nolimits_{i\ne j}
\E[(\omega_i e_i)\xi_i]'\Gamma_{ij}\xi_j\right|^2
\leq \frac{\|\Gamma\|^2}{\omega^2}
\frac{\sum_i\V(\omega_i e_i)}{\omega^2}\to0.
\]
Because the quadratic variance is non-negligible on these subsequences, both convergences follow from the same row and spectral-norm bounds used above to apply Lemma~\ref{lem: CLT for quadratics}. Thus the first condition also holds, and asymptotic normality for the joint linear-quadratic term follows: $x'Ae/\omega \xrightarrow{d} \mathcal{N}(0,1).$
\end{proof}

\begin{proof}[Proof of Lemma \ref{lem: assumption check}] For the first statement notice that $\sum_{\tilde\ell} (A^*_{\ell \tilde\ell})^2 = A^*_{\ell\ell} \leq 1$, so:
\[
\sum_{j=1}^N \| A_{ij}^*\|_F^2 \leq    \sum_{\ell \in S_i} A^*_{\ell\ell} \leq C \max_i T_i.
\]
According to Theorem \ref{thm: leave out}, $A^* = (I - B)M$, where $B$ is sparse with non-zero elements $B_{\ell\tilde\ell}$ only if   $\mathcal{E}_{\ell\tilde\ell}=0$. If $\|M_{\ell}^{-1}\| \leq C$, then $\|B_\ell\|^2 \leq C$ and all elements of $B$ are bounded. Thus, $\|I-B\|_\infty\leq C\max_i T_i$,
\begin{align*}
    &\max_i(A'\lambda)_i^\prime\Sigma_{e,i}(A'\lambda)_i 
    \leq C\max_i\|\Sigma_i\|\cdot\|(A'\lambda)_i\|_2^2
    \leq C \max_i\|\Sigma_i\|\cdot\max_iT_i\|(A'\lambda)\|_\infty^2 \\ 
    &\leq C \max_i\|\Sigma_i\|\cdot \max_iT_i\cdot\|M\|_\infty^2\|1-B\|_\infty^2\|\lambda\|_\infty^2
    \leq C \max_i\|\Sigma_i\|\cdot\max_iT_i^3\|M\|_\infty^2. 
\end{align*}
For the third statement:
\begin{align*}
    &\sum_\ell (A^*_{\ell\ell_1})^2\leq C\sum_\ell M^2_{\ell\ell_1}+C\sum_\ell(\sum_{\tilde\ell:\mathcal{E}_{\ell\tilde\ell}=0}B_{\ell\tilde\ell}M_{\tilde\ell \ell_1})^2\leq CM_{\ell_1,\ell_1}+C\sum_\ell\|B_\ell\|^2\sum_{\tilde\ell:\mathcal{E}_{\ell\tilde\ell}=0}M_{\tilde\ell \ell_1}^2 \\ 
    &\leq C+C\sum_\ell\sum_{\tilde\ell:\mathcal{E}_{\ell\tilde\ell}=0}M_{\tilde\ell\ell_1}^2\leq C+C(\max_iT_i)M_{\ell_1,\ell_1}\leq C\max_iT_i.
\end{align*}
Here we used the fact that $\mathcal{E}_{\ell\tilde\ell}=0$ only within a cluster and thus has less than $T_i$ elements. 
$$ \sum_{i=1}^N\|A_{ij}^*\|_F^2\leq \sum_{\ell_1\in S_j}\sum_{\ell=1}^n A_{\ell\ell_1}^*\leq C\max_kT_k^2. 
$$
From Theorem \ref{thm: leave out}, $B$ is block-diagonal with row/column sparsity bounded by $\max_i T_i$. Thus:
$\|B\|^2 \leq \|B\|_1 \|B\|_\infty \leq C (\max_i T_i)^2.$
Hence,    $\| A^*\|^2\leq \| (I-B)M\|^2\leq  C(\max_iT_i)^2.$
\end{proof}

\begin{proof}[Proof of Lemma \ref{lem: jackknife variance}.]
We start by expanding the jackknife difference:
$$
\mathcal{Z}-\mathcal{Z}_{(j)}=-\mu_j+\omega_je_j+\sum_{i\neq j}(\lambda_j'A^*_{ji}e_i+v_i'A^*_{ij}w_j'\delta)+\sum_{i\neq j}(v_i'A^*_{ij}e_j+v_j'A^*_{ji}e_i),
$$
where $\mu_j=\sum_{i\neq j}\lambda_i'A_{ij}^*w_j'\delta$, $\omega_j=(A^*\lambda)_j+v_j'A_{jj}^*$. We now compute the squared expectation:
\[
\E[(\mathcal{Z} - \mathcal{Z}_{(j)})^2] = \mu_j^2 + \V(\omega_j e_j) + \sum_{i \neq j} \V\left( \lambda_j' A^*_{ji} e_i + v_i' A^*_{ij} w_j' \delta \right) + \sum_{i \neq j} \V\left( v_i' A^*_{ij} e_j + v_j' A^*_{ji} e_i \right),
\]
using the fact that all components are uncorrelated due to independence across clusters.
Summing over $j$ and comparing to the true variance $\V(\mathcal{Z})$, we obtain:
\begin{align}\label{eq: jackknife var}
\E[\hat V_\textrm{JK}] = \V(\mathcal{Z}) + \sum_j \mu_j^2 
+ \sum_{i=1}^N \sum_{j < i} \V\left( v_i' A^*_{ij} e_j + v_j' A^*_{ji} e_i \right)
+ \sum_j \sum_{i \ne j} \V\left( \lambda_j' A^*_{ji} e_i + v_i' A^*_{ij} w_j' \delta \right).
\end{align}

Therefore, $\E[\hat V_\textrm{JK}] \geq \V(\mathcal{Z})$, and the jackknife variance estimator is conservative. In the special case $A^*_{ij} = 0$ for all $i \ne j$, all additional terms vanish, yielding
$\E[\hat V_\textrm{JK}] = \V(\mathcal{Z}).$
\end{proof}

\section{Simulation in Section~\ref{sec:design_base}}
\label{sec:sim}

The DGP is a stochastic block model with $N = 50$ clusters of size $T = 10$, for a total of $n = 500$ units. Within each cluster, each pair of units forms an undirected edge independently with
  probability $p_{\mathrm{in}} = 0.3$; there are no cross-cluster edges. The weighted adjacency matrix $G$ is defined as $G_{\ell k} = F_{\ell k} \cdot \omega_{\ell k}$, where $F_{\ell k} \in \{0,1\}$
  indicates whether $\ell$ and $k$ are connected and $\omega_{\ell k} = \omega_{k\ell} \overset{\mathrm{iid}}{\sim} \mathrm{Exp}(1)$ for each edge. Crucially, $G$ is \emph{not} row-normalized: the row
   sums $s_\ell = \sum_k G_{\ell k}$ vary across units. Cluster-level treatment saturation is drawn as $\mu_i \overset{\mathrm{iid}}{\sim} U[0.1, 0.9]$ for $i = 1, \dots, N$. The network structure, edge
  weights, and cluster-level saturations are sampled once and held fixed across all simulation replications. Unit-level treatment drawn as $x_\ell\sim \text{Bernoulli}(\mu_{i(\ell)})$ independently across units and the outcome is $y_\ell =     
  \beta  x_\ell + \alpha \sum_{k} G_{\ell k}  x_k + e_\ell$ where $e_\ell \overset{\mathrm{iid}}{\sim} N(0,1)$. $x_\ell, e_\ell$ are redrawn in each of $S = 2{,}000$ simulation replications and $\beta=1$. 

\section{Additional results for Section~\ref{sec: empirical application}}

\begin{figure}[htb!]
    \begin{center}\subfloat[Baseline\label{fig:App_panelA}]{
            \includegraphics[width=0.4\textwidth]{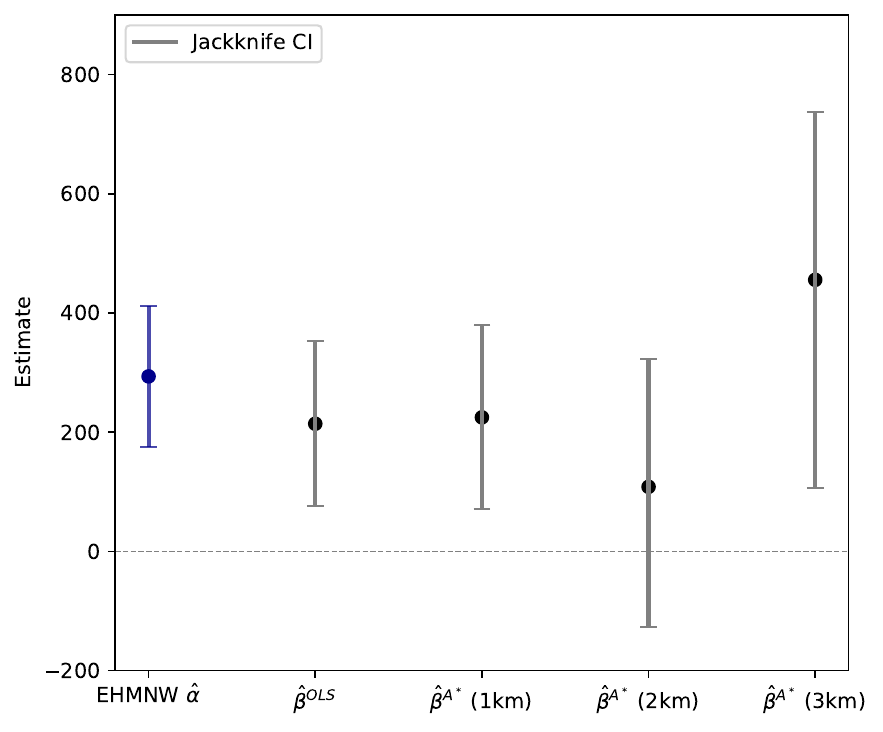}
        }
        \hfill
        \subfloat[Continuous treatment\label{fig:App_panelB}]{
            \includegraphics[width=0.4\textwidth]{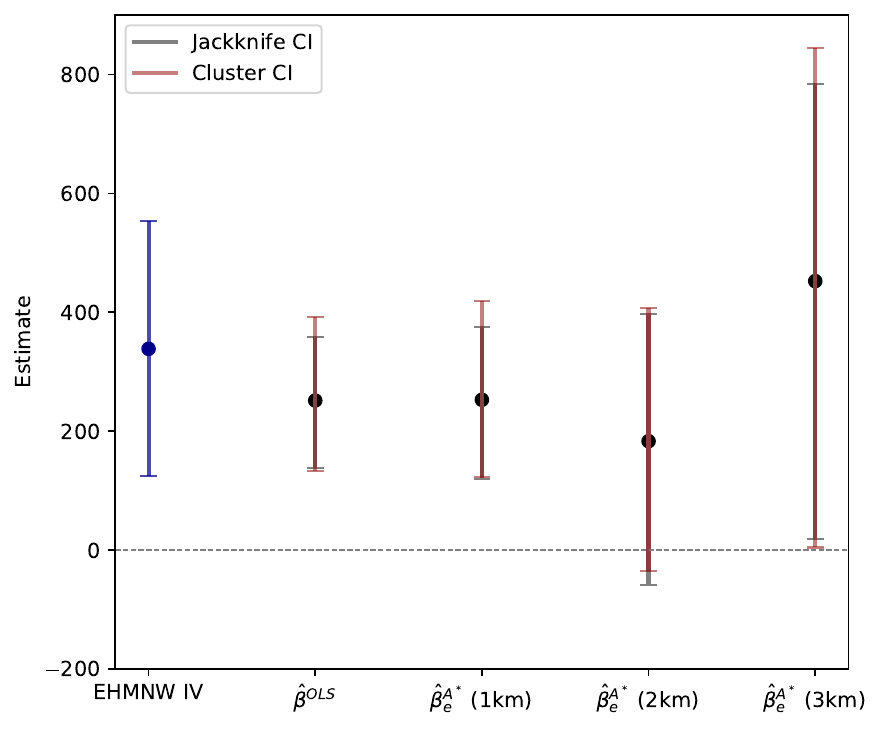}
        }
        \caption{Design-based estimates of Figure~\ref{fig:estimates}}
        \label{fig:design}
    \end{center}
    {\small Notes. These two panels reproduce the results in Figure~\ref{fig:estimates} using the design-based estimator. 
    }
\end{figure}

The main text reports the doubly robust construction. For comparison, this appendix reproduces the design-based estimates under the corresponding treatment assignment equations. The difference between the design-based and the doubly robust estimator is that the design-based $A^*$ matrix need not satisfy $AM=A$. Suppose the outcome depends on the fixed effect in significant way as well $y=\beta x+W\gamma+e$, then for one sided $A$ we have $\hat\beta^A-\beta = \frac{x'AW\gamma+x'Ae}{x'Ax}$. Under a correctly specified assignment model, $\E[x'AW\gamma]=\E[v'AW\gamma]=0$ but the term $x'AW\gamma$ is stochastically non-zero and the estimator is not invariant to value of $\gamma$. However, by imposing $AM=A$, the doubly robust estimator eliminates this term exactly in every sample. In our empirical setting, the design-based estimates exhibit variability at $R=2$ and $R=3$, whereas the doubly robust estimates remain stable. At the same time, the design-based estimates are not statistically different from the doubly robust estimates under paired Wald tests.

\end{document}